\documentclass[reprint, amsmath, amssymb, aps, prx, floatfix,]{revtex4-2}

\usepackage[T1]{fontenc}
\usepackage{newtxtext}
\usepackage{newtxmath}
\usepackage{calrsfs}
\DeclareMathAlphabet{\mathcal}{OMS}{cmsy}{m}{n}
\usepackage{amsmath}
\usepackage{amssymb}
\usepackage{rotating}
\usepackage{physics}
\usepackage[table,dvipsnames]{xcolor}
\usepackage{tcolorbox}
\usepackage{mathtools}
\usepackage{array}
\usepackage[utf8x]{inputenc}
\usepackage[normalem]{ulem}
\usepackage{amsbsy}
\usepackage{diagbox}

\usepackage{amsthm}
\newtheorem{proposition}{Proposition}
\newtheorem{theorem}{Theorem}
\newtheorem{corollary}{Corollary}
\theoremstyle{definition}
\newtheorem{definition}{Definition}
\theoremstyle{remark}
\newtheorem{remark}{Remark}

\usepackage{multirow}
\usepackage{makecell}
\usepackage{array}
\usepackage{hhline}
 
\renewcommand{\arraystretch}{1.4}
\newcommand{\yes}{\textcolor{ForestGreen}{\pmb{\checkmark}}}
\newcommand{\no}{\textcolor{Red}{\pmb{$\times$}}}

\usepackage{natbib}
\usepackage{graphicx}
\usepackage{dcolumn}
\usepackage{bm}
\usepackage{hyperref}
\usepackage{tikz}
\usepackage{circuitikz}

\usepackage{geometry}
\begin{document}

\title{General theory of monitored Quantum Reservoir Computing
}

\author{Oriol Morguí-Sancho}
\email{oriolmorgui@ifisc.uib-csic.es}
\author{Gonzalo Manzano}
\author{Gian Luca Giorgi}
\email{gianluca@ifisc.uib-csic.es}
\author{Roberta Zambrini}
\email{roberta@ifisc.uib-csic.es}
 \affiliation{Instituto de Física Interdisciplinar y Sistemas Complejos (IFISC), UIB–CSIC\\
UIB Campus, Palma de Mallorca, E-07122, Spain}

\date{\today}

\begin{abstract}

Quantum reservoir computing  (QRC) provides a powerful framework for processing temporal data using quantum dynamics, but incorporating measurements into the reservoir remains a fundamental challenge and distinctive feature with respect to classical settings. The induced back-action can vary from a source of disturbance to a computational resource, as measurement deeply modifies the dynamics underlying temporal processing. Existing approaches have treated specific monitoring schemes independently, missing the common physical principles governing online quantum reservoirs. Here we develop a general theory of monitored quantum reservoir computing based on indirect quantum measurements, which unifies projective, weak, partial, and dissipative monitoring protocols within a single operational framework. 
Measurement back-action can serve as a controllable resource, providing the effective dissipation and non-unital dynamics required for successful QRC, even when the underlying unmonitored evolution is unitary, a setting otherwise unsuitable for in-memory QRC. We derive general criteria under which monitored dynamics satisfy the echo-state property, fading memory, and input separability, including a necessary and sufficient condition for emergent strict contractivity.
By comparing different monitoring schemes under a common reference dynamics, we show that these protocols are not interchangeable parameterizations to be optimized for peak performance, but rather constitute qualitatively distinct routes to computational capability, each enabled by the interplay between information extraction and measurement-induced disturbance — a trade-off that can be further shaped through time multiplexing. Our results provide a unified theoretical foundation for online monitored quantum reservoir computing and establish quantum measurement engineering as a systematic approach for designing reservoir architectures across different quantum platforms.
\end{abstract}

\maketitle

\section{\label{sec:introduction} Introduction}

Quantum monitoring is the process of extracting information from a quantum system over time through sequential measurements \cite{Wiseman_Milburn_2009, Jacobs_2014}. These measurements disturb the system’s state through back-action, which has to be accounted for when describing the monitored system’s dynamics. Modeling and controlling the monitoring back-action is essential for applications where information about the system needs to be extracted in real-time, for instance, in repeated syndrome measurements for error correction \cite{PhysRevX.11.041058}, in metrology and sensing \cite{PhysRevX.13.031012, quantum-sensing}, and in time series processing in quantum machine learning \cite{mujal_weak}. In some of these applications, the goal is to minimize disturbance, for example by using weak measurements \cite{Murch_2013}. In other cases, measurement back-action can be harnessed as a resource to intentionally drive specific dynamics \cite{steering, entanglement-phase-trans, Bose,Koh_2023}. Engineering proper measurement schemes is thus crucial for online quantum protocols. 

In general, any type of measurement can be implemented as a partial measurement on an enlarged Hilbert space \cite{Wiseman_Milburn_2009, gelfand-neumark}, where the system of interest is unitarily coupled to an auxiliary quantum probe that is projectively measured. This approach enables the realization of generalized measurements beyond standard projective ones, including weak measurements, and provides a way to tailor the back-action. When complemented with a reset operation, the probe can be reinitialized after each measurement and reused, enabling repeated monitoring in a scalable manner. In qubit platforms, these generalized monitoring schemes can already be implemented \cite{Murch_2013, PhysRevA.90.032302, Brun_Diosi_unsharp}, for example, in superconducting devices, where fast and reliable reset operations have recently become available \cite{PhysRevLett.121.060502, hua2023exploitingqubitreusemidcircuit}.

A particularly natural setting where the interplay between back-action and information extraction plays a central role is quantum reservoir computing (QRC), a framework for processing temporal data using quantum dynamics  \cite{FujiiNakajima2017}. In QRC, an input sequence is encoded in the dynamics of a quantum system, called the reservoir, whose evolved state serves as a high-dimensional representation of the input history \cite{opportunities}. Measurements are required to read out input-dependent features, and these are then processed by a trained linear readout layer, while the reservoir itself remains fixed. Through its dynamics, the reservoir retains information about past inputs and temporal correlations, with its memory serving as a key computational resource  \cite{FujiiNakajima2017, NokkalaGaussian}.

Many existing QRC proposals focus on the reservoir dynamics but do not address monitoring and online measurement effects \cite{FujiiNakajima2017,NokkalaGaussian,opportunities}. To circumvent this aspect operationally, one can employ restarting or rewinding protocols \cite{mujal_weak, experimental-qrc}: after each measurement, the reservoir is re-initialized and then evolved again to incorporate the next input. This approach, however, precludes real-time processing and displays an evident overhead with respect to classical implementations where back-action can generally be neglected. There is therefore growing interest in developing online QRC architectures enabling experimental implementations by incorporating quantum measurement in the reservoir dynamics. Different proposals include weak \cite{mujal_weak}, projective \cite{chen-nurdin-yamamoto, yasuda}, or partial \cite{hu_et_al, connerty2024predictingchaoticsystemsquantum, scaling_laws} measurements in qubit-based systems, or by using feedback mechanisms, either coherent \cite{jorge-qrc}, incoherent \cite{feedback-qrc,Paparelle2026}, or hybrid \cite{feedback-midcircuit, feedback-weak}. While several approaches have addressed online monitoring strategies, a unified framework enabling their systematic comparison and task-dependent optimization is still lacking.

In this work, we develop a general framework for online monitored QRC and identify viable reservoir designs. 
The main contributions are as follows. We develop a general theory of monitored quantum reservoir computing based on the framework of indirect quantum measurements, providing a unified description of projective, weak, dissipative, and partial monitoring schemes. This formulation enables a classification of monitoring protocols according to their dynamical properties and establishes general conditions under which monitored reservoirs satisfy the echo-state property, fading memory, and input separability, identifying viable QRC maps.  Within this common framework, we perform a systematic comparison of different monitoring strategies under identical underlying reservoir dynamics, revealing their respective advantages and limitations, and also accounting for the diverse impact of sampling noise. Finally, we identify general design principles for monitored quantum reservoirs by elucidating the trade-off between information extraction and measurement-induced disturbance, and show how this balance can be further optimized through time multiplexing.

We begin by introducing the theoretical background for monitoring quantum systems via quantum measurements, and present several tunable measurement schemes that allow continuous control over the strength of the back-action (Sec. \ref{sec:monitoring-qubits}). We then formulate the general monitored QRC framework and show that measurement back-action can serve as a resource to induce the dynamical properties required for effective QRC, namely echo-state, fading memory, and separability properties (Sec. \ref{sec:monitored_qrc}). Our framework naturally leads to the classification summarized in Table \ref{table:monitored-qrc}. It allows for a systematic comparison of existing monitoring proposals by fixing a common unmonitored dynamics as a reference, addressing both Ising Hamiltonian and random unitaries. Finally, exploiting the inherently online character of the monitored scheme, we demonstrate that performance can be further enhanced through time-multiplexing, revealing a fundamental trade-off between measurement disturbance and the information extracted from the reservoir (Sec. \ref{sec:time-mult}).

\section{Quantum monitoring with qubits}\label{sec:monitoring-qubits}

A defining feature of quantum measurement is the coexistence of information extraction and state disturbance \cite{Wiseman_Milburn_2009, Jacobs_2014}. The former is characterized by assigning to each possible measurement outcome $m$ a positive semi-definite Hermitian operator $E_m$, called \textit{effect}, that determines the probability of obtaining that outcome
\begin{equation}
    \text{Pr}_\rho(m)=\text{Tr}[E_m\rho],    
\end{equation}
given that the measured system is in state $\rho$. The single requirement on the set of effects $\{E_m\}$ is to satisfy the completeness relation $\sum_m E_m = \mathbb{I}$. This set defines a positive operator-valued measurement (POVM), which corresponds to the most general measurements in quantum mechanics. It can be used to obtain information about an observable $X$ if all the effects commute with it, i.e., $[E_m, X]=0, \forall m$.

The measurement-induced state disturbance is referred to as \textit{back-action} and determines the post-measurement state. This back-action, for outcome $m$, is described by a completely positive trace-preserving (CPTP) map $\mathcal{M}^{(m)}$, such that the post-measurement state becomes
\begin{equation}
    \rho^{(m)} = \mathcal{M}^{(m)}\rho.
\end{equation}
We restrict ourselves to efficient measurements, where acquisition of information is guaranteed, and no classical uncertainties are present \cite{Wiseman_Milburn_2009,Jacobs_2014}. Then, the back-action is fully captured by a set of measurement (or Kraus) operators $\{\Omega_m\}$ satisfying $\Omega_m^\dagger \Omega_m =E_m$, and the measurement map is given by
\begin{equation}\label{eq:cond-meas-map}
    \mathcal{M}^{(m)}\rho = \frac{\Omega_m\rho \Omega_m^\dagger}{\text{Pr}_\rho(m)}.
\end{equation}
These operators, as well as the effects, need not be projectors and describe more general quantum measurements. Significant examples, encompassing different back-actions, will be presented in the next subsection.

In a monitored quantum system, measurements are performed sequentially, each  inducing an outcome-dependent disturbance. After $k$ measurements, the state evolution becomes conditioned on the measurement record $\vec{\mathbf{m}}_k = (m_k,\dots,m_1)$,
\begin{equation}\label{eq:q-traj}
\rho^{(m_k,\dots,m_1)}_k=\left(\mathcal{M}^{(m_k)}\circ\Lambda_k\right)\rho_{k-1}^{(m_{k-1},\dots,m_1)}.
\end{equation}
where $\rho_{k-1}^{(m_{k-1},\dots, m_1)}$ is the state after $k-1$ measurements. Here, $\Lambda_k$ is a generic CPTP map describing the system evolution between measurements $k-1$ and $k$. The stochastic evolution defined in Eq. (\ref{eq:q-traj}) represents one step of a single realization of the monitored dynamics (or quantum trajectory), as illustrated in Fig. \ref{fig:neumarks}a. Different runs of the experiment will generally produce different outcome sequences $\vec{\mathbf{m}}_k$ and, consequently, different state trajectories in Eq. \eqref{eq:q-traj}.

\begin{figure}
    \centering
    \includegraphics[width=1.0\linewidth]{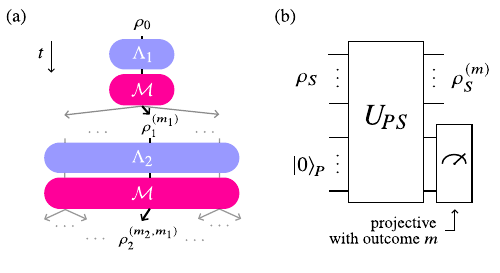}
    \caption{(a) The monitored dynamics of a quantum system can be viewed as the interleaved effect of unmonitored dynamics $\Lambda_k$, which evolve the system in a deterministic way, and measurements $\mathcal{M}$, which stochastically change the state of the system depending on the measurement outcome $m$.  A single realization of the stochastic dynamics, with measurement history $\vec{\mathbf{m}}_k=(m_k,\dots,m_1)$, is known as a quantum trajectory; (b) According to Gelfand-Naimark's theorem, the measurement map $\mathcal{M}$ can always be implemented as an indirect measurement by unitary coupling between the system of interest ($S$) and a quantum probe ($P$) followed by a projective measurement of the latter.}
    \label{fig:neumarks}
\end{figure}

The impact of measurement on the monitored state can vary significantly, ranging from a weak perturbation to a substantial alteration of the state. In particular, projective measurement back-action destroys coherence between the projected eigenstates. Alternatively, information can be extracted indirectly through an auxiliary system coupled to the state of interest \cite{Braginsky_Khalili_Thorne_1992}. In fact, any POVM can be modeled as an indirect extraction of information from an auxiliary system, which allows us to tune the back-action on the system. Such an indirect measurement corresponds to a partial projective measurement on an extended Hilbert space \cite{Wiseman_Milburn_2009} where the system of interest ($S$) is coupled to an auxiliary one, the probe ($P$), through a unitary interaction $U_{PS}$, followed by a projective measurement of the probe (see Fig. \ref{fig:neumarks}b). The post-measurement state of $S$ is then obtained by tracing out $P$, yielding the measurement map 
\begin{equation}\label{eq:gelfand-neumark-selective}
    \mathcal{M}^{(m)}\rho=\frac{1}{\text{Pr}_\rho(m)}\text{Tr}_P\left[(\Pi^{(m)}_P\!\otimes\mathbb{I}_S)U_{PS}\left(\ket{0}\!\!\bra{0}_P\otimes\rho\right)U_{PS}^\dagger\right],
\end{equation}
where $\Pi_P^{(m)}=\ket{\varphi_m}\!\!\bra{\varphi_m}_P$ denotes the probe projector associated to outcome $m$, and $\ket{0}_P$ is the probe's initial (ready-to-measure) state. According to the Gelfand-Naimark dilation theorem \cite{Wiseman_Milburn_2009, gelfand-neumark}, this construction describes any efficient quantum measurement and is equivalent to (\ref{eq:cond-meas-map}), with measurement operators
\begin{equation}\label{eq:meas-ops}
    \Omega_m = \bra{\varphi_m}_{P} U_{PS}\ket{0}_{\!P}.
\end{equation}
By appropriately engineering the coupling unitary $U_{PS}$ and the probe measurement basis $\{\ket{\varphi_m}\}$, one gains full control over both the back-action imparted to $S$ and the information extracted from it.

The monitored dynamics described by Eq. (\ref{eq:q-traj}) requires each measurement to be implemented by coupling the system with a fresh probe, which can be either a single auxiliary system that is reset and reused at each step or a sufficiently large auxiliary system where the entire measurement record can be stored. This structure is closely related to collision models for open quantum dynamics \cite{collision-models, Brun_trajectories}, and can be extended beyond the efficient measurements case \cite{Manzano18}. Indeed, the non-selective measurement map, obtained by averaging over all measurement outcomes,
\begin{align}\label{eq:nonselective-measurement}
    \mathcal{M}\rho&=\sum_m \text{Pr}_\rho(m)\cdot\mathcal{M}^{(m)}\rho \nonumber\\
    &=\sum_m \Omega_m\rho\Omega^\dagger_m,
\end{align}
takes the form of a Kraus representation of a CPTP map describing deterministic but irreversible dynamics. The non-selective monitored evolution is therefore analogous to open quantum system dynamics, with the key distinction that the environment (here, the probe) is measured to extract information about the system. In the following, we show that these maps allow us to identify measurement schemes with controllable back-action, where the dissipation rate translates directly into measurement strength, and provide explicit operational examples.

\begin{figure*}
    \centering
    \includegraphics[width=1.0\linewidth]{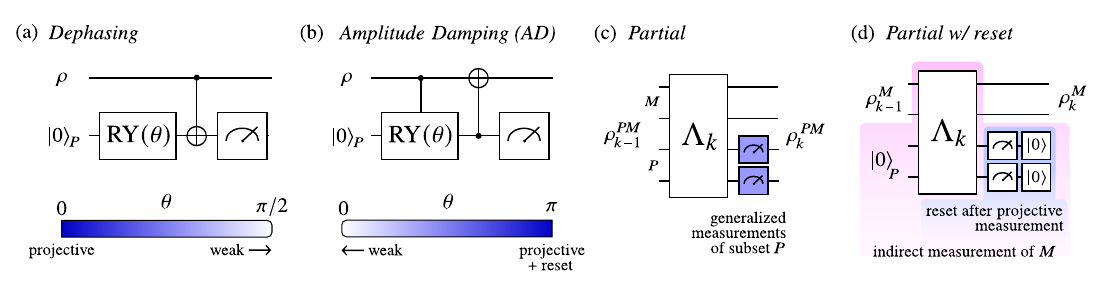}
    \caption{Using the equivalence between open quantum dynamics and nonselective measurement maps, one can engineer indirect measurement schemes whose average back-action reproduces dissipative channels such as (a) dephasing and (b) amplitude damping. In both cases, a system qubit is coupled to a probe qubit through the unitary circuits shown, with a tunable parameter $\theta$ controlling the measurement strength from weak to strong. For many-body systems, the measurement-induced disturbance on the natural dynamics $\Lambda_k$ can also be reduced by measuring only a subset of qubits $P$, which is equivalent to indirectly measuring the complementary subset $M$ (c,d). Sequential measurements (monitoring) without resetting $P$ (c) lead to non-Markovian dynamics, since $P$ retains memory of previous outcomes. Resetting $P$ after each measurement (d) restores the Markovian monitoring protocol discussed in the main text. In these schemes, the effective measurement strength is mainly determined by the size of the measured subset $P$.}
    \label{fig:fin-strength}
\end{figure*}

\subsection{Controllable-strength measurements with quantum circuits}\label{subsec:indirect-meas}

Monitoring entails a fundamental trade-off between information gain and state disturbance \cite{fuchs-jacobs}, which becomes particularly relevant when a finite number of experimental runs is available. Projective measurements extract the maximum amount of information from a quantum system at the cost of strong disturbance, and tuning the measurement strength allows one to trade off information gain for disturbance continuously. Controllable-strength measurements can be implemented as indirect measurements.

In this work, we focus on a concrete and experimentally relevant setting in which both the system and the probe consist of qubits, and the coupling unitary can be represented as a quantum circuit. Each system qubit is coupled to its own probe through an identical interaction unitary $U_{PS}$, and projections on the probe are in the computational basis. In this case, the system measurement map is characterized by the $N$-qubit operators
\begin{equation}\label{eq:n-qubit-meas}
    \Omega_{\mathbf{b}_m}=\Omega_{b_1}\otimes\cdots\otimes\Omega_{b_N},
\end{equation}
where $\mathbf{b}_m=b_1\cdots b_N$ is a bit-string with $b_i$ representing the outcome corresponding to qubit $i$.

\subsubsection{Dephasing measurement map}

A tunable-strength measurement scheme, whose averaged back-action is equivalent to a dephasing channel, can be constructed by setting the coupling unitary as \cite{Brun_Diosi_unsharp}
\begin{equation}\label{eq:unitary-deph}
    U_{PS}^\text{(deph)}=\mathrm{CNOT}_{S,P}\left(\mathrm{RY}(\theta)_P\otimes\mathbb{I}_S\right), \quad \theta\in\left[0,\pi/2\right),\end{equation}
where the first and second indices of the controlled gate denote the control and target qubits, respectively. The corresponding measurement operators are
\begin{align}\label{eq:deph-meas-ops}
    \Omega_{0}(\theta) &= \cos\left({\theta}/{2}\right)\Pi_0 + \sin\left({\theta}/{2}\right)\Pi_1 , \nonumber \\
    \Omega_{1}(\theta) &= \sin\left({\theta}/{2}\right)\Pi_0 + \cos\left({\theta}/{2}\right)\Pi_1,
\end{align}
where $\Pi_0=\ket{0}\!\!\bra{0}$ and $\Pi_1=\ket{1}\!\!\bra{1}$. Consequently, the outcome probabilities are $\text{Pr}_\rho(0)=p\rho_{00} + (1-p)\rho_{11}$ and $\text{Pr}_\rho(1) = (1-p)\rho_{00} + p\rho_{11}$, where $\rho_{00},\rho_{11}$ are the populations and $p=\cos^2(\theta/2)$. At $\theta = 0$, this reduces to a strong projective measurement in the computational basis, with a completely dephasing back-action. For $0<\theta<\pi/2$, the measurement is of finite strength, with back-action decreasing as $\theta$ increases at the cost of extracting less information per measurement. While finite-strength measurements are often referred to as weak, we will reserve this term for measurements with $\theta$ close to $\pi/2$, where both the state disturbance and the information extracted per measurement are small (Fig. \ref{fig:fin-strength}a).

The dephasing measurement scheme defined above constitutes a measurement of $\sigma^z$, since all the effects $\{E_0 = p\Pi_0 + (1-p)\Pi_1, E_1 = (1-p)\Pi_0+p\Pi_1\}$ commute with this observable. Moreover, it describes a quantum non-demolition (QND) measurement of the same observable, since the expectation values of $\sigma^z$ are unaffected by the measurement \cite{qnd_braginsky}. This expectation value can be estimated from the outcome statistics as \cite{contextual-values, contextual-values-2} 
\begin{equation}\label{eq:deph-meas-sigma}
    \langle\sigma^z\rangle=\frac{1}{\cos\theta}\left(\text{Pr}_\rho(0)
    - \text{Pr}_\rho(1)\right).
\end{equation}
This ideal expectation value becomes actually noisy when considering finite ensemble size, as in any realistic implementation. Concretely, for finite-strength measurements, the statistical uncertainty in $\langle\sigma^z\rangle$ is larger than for projective measurements of the same observable, requiring more experimental runs to achieve the same precision (see Appendix \ref{app:finite}). In the limit of infinitely many experimental shots, however, the precision of finite-strength measurements coincides with that of projective ones.

When $N$ qubits are measured in the same way, higher-order correlations ($\sigma^z\otimes\sigma^z$,...) can also be extracted from the outcome statistics. The statistics of $\sigma^x$ and $\sigma^y$ can likewise be accessed by rotating the system state before measurement and inverting the rotation afterwards \cite{mujal_weak}.

\subsubsection{Amplitude Damping measurement map}\label{subsec:AD}

An unraveling for the Amplitude Damping (AD) channel can be constructed by coupling the system qubit to an auxiliary qubit through \cite{Nielsen_Chuang_2010}
\begin{equation}
    U_{PS}^\text{(AD)}=\text{CNOT}_{P,S}\text{CRY}(\theta)_{S,P} \,,\quad\theta\in (0,\pi],
\end{equation}
where $\theta$ controls the damping strength. The corresponding map on the system, after projective measurement on the probe qubit (Fig. \ref{fig:fin-strength}b), yields system measurement operators
\begin{align}\label{eq:ad-meas-ops}
    \Omega_{0}(\theta) &= \Pi_0 + \cos\left({\theta}/{2}\right)\Pi_1, \nonumber \\
    \Omega_{1}(\theta) &= \sin\left({\theta}/{2}\right)\ket{0}\!\!\bra{1},
\end{align}
and outcome probabilities $\text{Pr}_\rho(0) = \rho_{00} + p\rho_{11}$ and $\text{Pr}_\rho(1) = (1-p)\rho_{11}$, with $p=\cos^2(\theta/2)$. Now, at $\theta=\pi$, the post-measurement state is always set to $\ket{0}_S$. For $0<\theta<\pi$, one obtains a finite-strength measurement, with the weak regime corresponding to $\theta$ close to 0. The nature of this weak measurement differs fundamentally from that of the dephasing scheme: whereas in the latter all outcomes produce little back-action, in the AD scheme, outcome $m=0$ is obtained with high probability and imparts a small disturbance, while outcome $m=1$ though rare, produces strong back-action and is not a weak perturbation of the state \cite{Brun_weak_measurements_universal}, leading to a detectable quantum jump \cite{Minev19,Chen21}.

Like the dephasing scheme, the AD measurement is compatible with $\sigma^z$, but it is not QND, since the measurement alters the populations of the system and therefore disturbs the observable itself. The expectation value of $\sigma^z$ can also be estimated from the outcome statistics as
\begin{equation}\label{eq:error-sigmaz-ad}
    \langle\sigma^z\rangle=\frac{1}{\sin^2(\theta/2)}\left[\text{Pr}_\rho(0) - \text{Pr}_\rho(1)\right]-\cot^2(\theta/2).
\end{equation}
As in the dephasing case, the statistical uncertainty in $\sigma^z$ is larger in the finite-strength regime than for projective measurements (see Appendix \ref{app:finite}), resulting in less information extracted per shot. Higher-order correlations such as $\sigma^z\otimes\sigma^z$ can likewise be extracted from the $N$-qubit AD measurement [Eq. (\ref{eq:n-qubit-meas})].

\subsubsection{Partial measurement map and reset}\label{subsec:partial}

Instead of the controllable-strength monitoring via tunable coupling to auxiliary units, as in the two previous examples, one can assign to some of the units of a large system (e.g., a part of a quantum circuit or a many-body system) the role of the probe. The system is partitioned into two subsets: the qubits that are measured ($P$) and those that are not ($M$). Any generalized measurement, which can be implemented indirectly, can be performed on the qubits of $P$, and the $N$-qubit measurement operators take the form
\begin{equation}
    \Omega^{PM}_{\mathbf{b}_m} = \Omega_{b_1}\otimes\cdots\otimes\Omega_{b_{N_P}} \otimes \mathbb{I}^{\otimes (N-N_P)},
\end{equation}
where $N_P$ is the number of probed qubits $P$ (Fig. \ref{fig:fin-strength}c). The back-action can thus be controlled by adjusting either the size of $P$ or the strength of the measurement performed on it. A particularly relevant case arises when the measurements on the probed qubits are projective. The partial measurement then disentangles the two subsystems, so that the post-measurement state factorizes as $\rho^P\otimes\rho^M$. If, in addition, the measurements are in the computational basis, the outcome statistics give access to $\langle\sigma^z\rangle$, $\langle\sigma^z\otimes\sigma^z\rangle$, and higher-order correlations for qubits in $P$.

The partial measurement scheme defined above differs from the sequential monitoring framework of the previous Section in two main aspects. First, the disturbance of the partial monitoring is governed by the natural $P$-$M$ interactions, instead of engineered couplings with extra probes, and the (not necessarily weak) back-action on the system becomes sensitive to the quantum circuit or many-body system interactions. Second, the reduced dynamics cannot be represented by a divisible map  (\ref{eq:q-traj}). More specifically, the fresh-probe assumption is no longer satisfied. The subsystem $P$ now plays the role of the probe for $M$, but its state generally depends on the outcome of previous measurements. Consequently, the dynamics of $M$ cannot be described by a sequence of CPTP maps like in Eq. (\ref{eq:q-traj}).

Interestingly, one can recover a collisional model dynamics, as in the dephasing and AD measurement maps, if the qubits in $P$ are \textit{reset} after measurement (Fig. \ref{fig:fin-strength}d). The reset operation erases all information stored in the probes from previous interactions, effectively restoring them as fresh probes, while $M$ retains memory of the past. This not only allows casting the partial measurement map as a collisional model, but also makes projective measurements the appropriate choice. Indeed, as the main back-action is given by the reset, in this case, the measurement of $P$ can be projective without concern for uncontrolled disturbance. In addition, this scheme can represent non-local measurements of the subset $M$, since a one-to-one coupling between qubits in $M$ and $P$ is not required. Finally, we note that given a QRC map $\Lambda_k$ (that does not need to be unitary), the monitoring with partial measurement with reset describes an inefficient measurement in general,  as the $P$-$M$ joint evolution $\Lambda_k$ (or coupling) is not assumed to be unitary.\\

The monitoring schemes described in this section are suitable for online quantum information-processing protocols, allowing in principle to reduce the system disturbance  while still extracting useful information. In the following, we frame the practical implementation of online monitored QRC based on indirect and partial measurements on quantum hardware, addressing different monitoring designs, establishing their relation, and identifying advantages \cite{chen-nurdin-yamamoto, mujal_weak, yasuda}.

\section{Qubit--Based Monitored Quantum Reservoir Computing}\label{sec:monitored_qrc}

\subsection{Fundamentals of Quantum Reservoir Computing}

Quantum reservoir computing \cite{FujiiNakajima2017} is a supervised machine learning paradigm designed for temporal data processing, whose objective is to learn a map between an input sequence (here classical data) $\mathbf{s} = \{s_k\}$ and a target sequence $\mathbf{y} = \{y_k\}$ (Fig. \ref{fig:qrc-paradigm}a). The input is sequentially encoded into a quantum dynamical system, called the reservoir, through a quantum channel $\mathcal{T}(s_k)$, so that the reservoir state after $k$ input injections is
\begin{equation}\label{eq:qrc-map}
    \rho_k = \mathcal{T}(s_k)\circ\cdots\circ\mathcal{T}(s_1)\rho_0.
\end{equation}
Due to this sequential input-dependent evolution, the state of the reservoir at any given time depends on the history of past inputs. We stress that here we focus on quantum \textit{in-memory} processing, where coherent memory emerges from the dynamics of the full quantum system. This is different from designs based on incoherent feedback, where  the system is fully (projectively) measured and reset at each time step, and  memory arises through incoherent feedback of measurements. These settings, limited to incoherent memory, have been addressed in different platforms \cite{feedback-qrc,Paparelle2026,Gyurik}.

In quantum in-memory processing, the channel composition \eqref{eq:qrc-map} governs the state evolution: for each step $k$, a set of measurements of the reservoir state is performed ($\{\Omega_i\}$), and the outcome statistics are used to construct a readout feature vector $\mathbf{x}_k=\{x_{k,i}\}$, where $x_{k,i}=\text{Pr}_{\rho_k}(i)$. Equivalently, when the measurements are compatible with a set of observables $\{O_i\}$, their expectation values can be used as features in place of the outcome probabilities. A constant bias term $x_{k,0}=1$ is typically appended to each feature vector.

The feature vectors are then fed to a readout layer, consisting of a linear function, to produce an output $\hat{{y}}_k$ as
\begin{equation}\label{eq:output}
    \hat{{y}}_k = \mathbf{W}\cdot\mathbf{x}_k,
\end{equation}
where $\mathbf{W}$ is the so-called \textit{weight vector}. This vector is the only trainable parameter in QRC. By contrast, the reservoir's internal parameters are kept fixed throughout the protocol. The training then reduces to optimizing the classical output layer via linear regression, making QRC significantly more efficient to train than other supervised machine learning paradigms.

To assess QRC performance, several benchmark tasks have been proposed \cite{benchmarks}. In this work, we consider the linear short-term memory task (STM) and the NARMA$n$ task (see Appendix \ref{app:benchmark-reservoir} for details). Both tasks are standard (for the sake of comparison across different settings) and probe the capacity of the reservoir to recall past inputs linearly and nonlinearly, respectively. Performance is quantified via the squared Pearson correlation coefficient between the target sequence $\mathbf{y}$ and the predicted sequence $\hat{\mathbf{y}}$ (Appendix \ref{app:benchmark-reservoir}).

\begin{figure}
    \centering
    \includegraphics[width=1.0\linewidth]{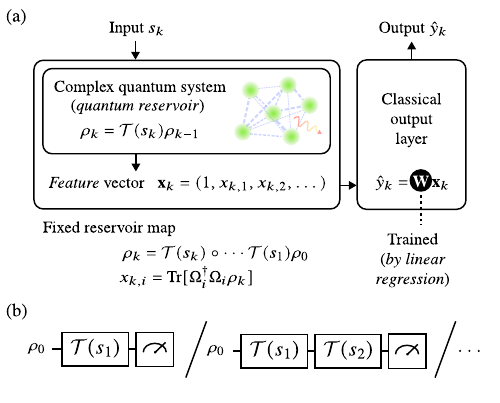}
    \caption{(a) QRC paradigm. At each step, an input $s_k$ is encoded in the state of a quantum system (the reservoir) through a completely positive trace-preserving map $\mathcal{T}$. The reservoir state is measured to extract a vector of features, $\mathbf{x}_k$, depending on both the current and past inputs. These features are then fed to a classical output layer in order to extract an output $\hat{y}_k$. The output weight vector $\mathbf{W}$ is the only trainable component, optimized via linear regression, while the reservoir parameters and dynamics remain fixed throughout the computation. (b) Restarting protocol. Since measurement back-action perturbs the reservoir dynamics, one way to circumvent this effect is to re-initialize the reservoir after each measurement. Starting from a fixed initial state, the reservoir evolves up to a given time step, where a measurement is performed. The reservoir is then reset to its initial state, and the evolution is repeated from the beginning to generate the output at the subsequent time step. This procedure is iterated for every step in the input sequence.}
    \label{fig:qrc-paradigm}
\end{figure}

Not every quantum dynamical system is suitable for QRC. The underlying dynamical map $\mathcal{T}(s_k)$ must satisfy three key properties: the echo-state property (ESP), the fading memory property (FMP), and input separability. These properties define universal (classical) reservoir computing and have also been studied in quantum settings \cite{NokkalaGaussian, chen-nurdin-yamamoto, universality-feedback, sannia}. The ESP ensures that, after a sufficiently long transient, the reservoir state becomes independent of its initial state. The FMP \cite{fading-memory} requires that the influence of past inputs diminishes over time, i.e., the reservoir gradually forgets inputs in the distant past. 

In finite-dimensional quantum systems, both properties are guaranteed if the dynamical map $\mathcal{T}(s_k)$ is strictly contractive for all inputs $s_k$ \cite{sannia, qrc_finite_dimensions}, i.e., if, for every input and every pair of states $\rho, \sigma$,
\begin{equation}\label{eq:strict-contractivity}
    \norm{\mathcal{T}(s_k)\rho-\mathcal{T}(s_k)\sigma}_1\leq \kappa\norm{\rho-\sigma}_1 \text{ with } 0\leq\kappa<1,
\end{equation}
where $\norm{\cdot}_1$ is the trace norm. Equivalently, one can define the \textit{contraction coefficient} as
\begin{equation}\label{eq:cont-coef}
    \kappa^\text{max}:=\sup_{\rho\neq\sigma}\frac{\norm{\mathcal{T}(s_k)\rho-\mathcal{T}(s_k)\sigma}_1}{\norm{\rho-\sigma}_1}.
\end{equation}
For any CPTP map, $\kappa^\text{max}\in[0,1]$, with distances between states either preserved or contracted. A strictly contractive map is one for which $\kappa^\text{max}\in[0,1)$, so that distances are uniformly contracted, a condition that guarantees the existence of a unique fixed point \cite{raginsky}. Then, regardless of the initial state $\rho_0$, the reservoir state (\ref{eq:qrc-map}) converges to this fixed point over a characteristic timescale $\tau_M$, which sets the fading memory of the QRC map. The strict contractivity condition excludes unitary dynamics as viable QRC maps since they preserve distances between states ($\kappa=\kappa^\text{max}=1$ for all pairs $\rho,\sigma$), requiring dissipation to be present to satisfy both the ESP and FMP \cite{sannia, qrc_finite_dimensions}.

In qubit-based QRC architectures, dissipation is typically introduced via the deterministic reset of part of the reservoir (erase-and-write maps \cite{FujiiNakajima2017}), natural or engineered noise (\textit{noisy channels} \cite{suzuki_natural, kubota_noise, domingo_noise_qrc, palacios-coherences}), continuous dissipative reservoir modeled by master equations \cite{Ghosh2019,sannia}, or as stochastic reset of part or all of the system (\textit{probabilistic reset} maps \cite{chen-nurdin-yamamoto, Fry2023}). As we will show in the following subsections, measurements on the reservoir can themselves introduce the dissipation required for suitable QRC dynamics, enabling monitored unitary evolution to serve as a viable QRC map.

In addition to these dynamical properties, a suitable reservoir must also ensure input separability, i.e., that different input series remain distinguishable in the long-time limit by producing distinguishable reservoir states. Together with the ESP, this implies that the unique fixed point of the dynamics (\ref{eq:qrc-map}) must depend on the input history within the fading memory of the reservoir. The combination of these properties rules out unital quantum maps (those satisfying $\mathcal{T}[\mathbb{I}/d]=\mathbb{I}/d$, where $d$ is the dimension of the Hilbert space) as suitable reservoirs in the long-time limit \cite{hu_et_al, rodrigo_input}, since the maximally mixed state becomes the unique fixed point, which is input-independent. However, non-unitality is not a sufficient condition for an input-dependent fixed point: the fully amplitude-damping channel provides a counterexample, as its unique fixed point is the (input-independent) state $\ket{0}\!\!\bra{0}$. Taken together, the ESP, FMP, and input separability thus require the QRC dynamics to be strictly contractive (sufficient condition for the ESP and FMP) and non-unital (necessary, but not sufficient, condition for input separability).

\subsection{Online monitored QRC map}\label{sub:online-monitored-qrc}

Many seminal works on QRC did not incorporate measurement back-action into the reservoir map $\mathcal{T}$ \cite{opportunities}. These proposals can be realized experimentally via  restarting protocols where the dynamics is   re-initialized after each measurement, re-injecting every time the extra input in the sequence to be processed \cite{mujal_weak} (see Fig. \ref{fig:qrc-paradigm}b). The efficiency can be improved by repeating only a shorter part of the dynamics (lasting the fading memory time of the QRC), in a rewinding protocol \cite{mujal_weak, experimental-qrc, memory-restriction-qrc, rewinding-exp}. 
These approaches, however, preclude a real-time implementation of the QRC architecture. Looking for online monitoring, the use of weak measurements  was first proposed in Ref. \cite{mujal_weak}. It was also proposed how to monitor QRC in photonic implementations via beam splitters \cite{jorge-qrc, jorge2}. Recently, several works have addressed real-time online processing in qubit-based reservoirs by incorporating back-action into the QRC map, with different measurement strategies -- indirect projective measurements \cite{chen-nurdin-yamamoto, yasuda}, partial measurements with reset \cite{hu_et_al, connerty2024predictingchaoticsystemsquantum, scaling_laws}, partial projective measurements \cite{ricci_controlled_damping, scaling_laws}, and (weak) dephasing measurements \cite{mujal_weak, memorynonlinearitytradeoffquantumreservoir, franceschetto} --, but a unified framework for systematic comparison has remained lacking. 

\begin{table*}[t]
\centering
\footnotesize
\begin{tabular}{|c|c|c|c||c|c|c|c|c|}
\cline{3-9}
\multicolumn{1}{c}{}
  & \multicolumn{1}{c|}{}
  & \multicolumn{1}{c|}{}
  & \multicolumn{1}{c||}{}
  & \multicolumn{5}{c|}{Dynamical map ($\Lambda$)} \\
\hhline{~~|~|~|-|-|-|-|-|}
    \multicolumn{1}{c}{}
  & \multicolumn{1}{c|}{}
  & \multicolumn{1}{c|}{\makecell{Strict\\contractivity}}
  & \multicolumn{1}{c||}{\makecell{Unitality}}
  & \multirow{2}{*}{Unitary}
  & \multirow{2}{*}{\makecell{Erase-\\and-write}}
  & \multirow{2}{*}{\makecell{Probabilistic\\reset}}
  & \multirow{2}{*}{\makecell{Unital\\noise}}
  & \multirow{2}{*}{\makecell{Non-unital\\noise}} \\
\hhline{~~|~|~|~|~|~|~|~|}
\multicolumn{1}{c}{}
  & \multicolumn{1}{c|}{}
  & \multicolumn{1}{c|}{}
  & \multicolumn{1}{c||}{}
  &
  &
  &
  &
  & \\
\hline
\hline
\multirow{4}{*}[-2ex]{\begin{turn}{90}Monitored ($\mathcal{M}$)\end{turn}
}
  & Dephasing
  & No & Yes
  & \no
  & \makecell{\yes\\ \cite{mujal_weak, memorynonlinearitytradeoffquantumreservoir, franceschetto}}
  & \makecell{\yes\\ \cite{chen-nurdin-yamamoto}}
  & \no
  & \makecell{\yes\\ \cite{yasuda, monzani2025nonunitalnoisesuperconductingquantum, ricci_controlled_damping}} \\
\cline{2-9}
  & AD
  & Yes & No
  &  \makecell{\yes\\  \cite{connerty2026controllablequantummemorycapacity}}
  & \yes
  & \yes
  & \yes
  & \makecell{\yes\\ \cite{sannia,sreetama}} \\
\cline{2-9}
  & \makecell{Partial\\projective}
  & No & Yes
  & \makecell{\no\\ \cite{hu_et_al, scaling_laws}}
  & \yes
  & \yes
  & \no
  & \makecell{\yes\\ \cite{monzani2025nonunitalnoisesuperconductingquantum}} \\
\cline{2-9}
  & \makecell{Partial with\\reset}
  & No & No
  & \makecell{\yes\\ \cite{hu_et_al, connerty2024predictingchaoticsystemsquantum, scaling_laws}}
  & \makecell{\yes\\ \cite{hu_et_al}}
  & \yes
  & \yes
  & \yes \\
\hline
\hline
\multicolumn{2}{|c|}{Unmonitored}
  & -- & --
  & \no
  & \makecell{\yes\\ \cite{FujiiNakajima2017}}
  & \makecell{\yes\\ \cite{chen-nurdin-yamamoto, Fry2023}}
  & \no
  & \makecell{\yes \\ \cite{suzuki_natural, kubota_noise, domingo_noise_qrc}}\\
\hline
\end{tabular}
\caption{Viability of the online QRC maps with different measurement schemes $\mathcal{M}$ and dynamics $\Lambda$. In the table, $\yes$ and \textcolor{Red}{\pmb{$\times$}} indicate whether the corresponding map is or is not a viable QRC scheme, respectively, while "Yes" and "No" denote whether a given measurement scheme satisfies the properties of strict contractivity and unitality. In this paper, we analyze in detail the monitored unitary and monitored erase-and-write maps with all the measurement schemes shown in the table.}
\label{table:monitored-qrc}
\end{table*}

Here, we formulate a general online monitored QRC framework in qubit-based architectures, building on the indirect measurement formalism introduced in Section \ref{sec:monitoring-qubits}. This theory naturally leads to the unified classification summarized in Table \ref{table:monitored-qrc}, which organizes monitored reservoir architectures according to the dynamical properties of both the monitoring channel and the underlying reservoir dynamics, thereby identifying the combinations that constitute viable quantum reservoirs. The selective evolution associated with a measurement record $\vec{\mathbf{m}}_k$ follows directly from Eq. \eqref{eq:q-traj} by identifying each dynamical step with an input-dependent monitored map $ \mathcal{T}^{(m)}(s) = ( \mathcal{M}^{(m)}\circ\Lambda)(s)$. The input sequence may be encoded through the reservoir evolution $\Lambda(s)$ \cite{FujiiNakajima2017,sannia,sannia-skin}, the measurement map \cite{jorge2,hu_et_al}, or both.

As in standard QRC, the readout is constructed from repeated realizations of the monitored dynamics. Averaging over the measurement outcomes yields the ensemble evolution of Eq. \eqref{eq:qrc-map}, where the QRC map is now given by $\mathcal{T}(s) = \left(\mathcal{M}\circ\Lambda\right)(s)$, with $\mathcal{M}$ the non-selective measurement map [Eq. \eqref{eq:nonselective-measurement}]. The feature vector is obtained either from the measurement outcome probabilities or, equivalently, from expectation values of compatible observables. In the former case,
\begin{equation}\label{eq:output-features}
    x_{k,i}
    =
    \mathrm{Tr}
    \!\left[
    \Omega_i^\dagger\Omega_i
    \Lambda(s_k)\rho_{k-1}
    \right].
\end{equation}

For the monitored dynamics to constitute a valid quantum reservoir, the resulting map $\mathcal{T}$ must satisfy ESP, FMP, and input separability discussed in the previous section. The ESP and FMP are guaranteed whenever $\mathcal{T}$ is strictly contractive. In particular, if either $\Lambda$ or $\mathcal{M}$ is strictly contractive, then
\begin{equation}\label{eq:strict-composition}
    \norm{(\mathcal{M}\circ\Lambda)(s_k)\rho-
    (\mathcal{M}\circ\Lambda)(s_k)\sigma}_1
    \le
    \kappa_{\mathcal{M}}^{\rm max}
    \kappa_{\Lambda}^{\rm max}
    \norm{\rho-\sigma}_1,
\end{equation}
for every input $s_k$ and every pair of states $\rho,\sigma$, with $\kappa_{\mathcal{M}}^{\rm max}\kappa_{\Lambda}^{\rm max}<1$, and the total map becomes strictly contractive. Interestingly, strict contractivity may emerge even when neither constituent map is individually strictly contractive (Theorem \ref{thm:disjoint} in Appendix \ref{app:strict} provides a necessary and sufficient criterion), as exemplified by the composition of dephasing channels with different dephasing axes or the repeated composition of random unitary dynamics with dephasing. This example will play an important role in the following.

Similarly, input separability fails for unital dynamics $\mathcal{T}$ in the long term, ruling them out as suitable QRC maps. For monitored reservoirs, if both $\Lambda$ and $\mathcal{M}$ are unital, so is the composed map. As with strict contractivity, non-unitality of $\mathcal{T}$ is not guaranteed even if both constituent maps are non-unital \footnote{An example of this is the concatenation of an amplitude damping channel with full damping, i.e. $\Lambda_1\rho = \, |0\rangle\!\langle 0|$, and a channel of the form $\Lambda_2\rho = (1-p)\rho + p \, \sigma_p$ with $(d-1)/d \leq p < 1$, where $\sigma_p = \mathbb{I}/(p d) + \left(1 - 1/p\right)|0\rangle\!\langle 0|$ is a valid quantum state in this parameter range and $d$ is the dimension of the Hilbert space. For any allowed value of $p$, the composition $\Lambda_2 \circ \Lambda_1$ is unital, $\Lambda_2 \circ \Lambda_1 (\mathbb{I}/d) = \mathbb{I}/d$, while neither sub-map is: $\Lambda_1(\mathbb{I}/d) = \, |0\rangle\!\langle 0|$ and $\Lambda_2(\mathbb{I}/d) \neq \mathbb{I}/d$ for $p < 1$.}.

A necessary, but not sufficient, condition for a viable monitored reservoir is therefore the non-unitality of at least one of the two maps. In particular, unital unmonitored dynamics $\Lambda$ must be complemented with non-unital measurement schemes. These observations open the possibility of engineering measurement schemes that endow otherwise unsuitable unmonitored dynamics, such as unitary maps \cite{luca}, with the ESP and input separability, providing a new route for the design of viable QRC reservoirs.

Table \ref{table:monitored-qrc} provides a unified classification of monitored QRC architectures, identifying which combinations of monitoring schemes and reservoir dynamics constitute viable quantum reservoirs. Representative realizations from the literature are included to position existing approaches within this general theoretical framework.

For the sake of clarity, we distinguish the intrinsic properties of the monitoring channels (strict contractivity and unitality) from the viability of the resulting QRC map, considering both standalone (unmonitored) and monitored reservoir dynamics $\Lambda$. The monitoring schemes included in Table \ref{table:monitored-qrc} encompass the representative classes introduced  above, namely dephasing, amplitude damping, partial projective measurements, and partial measurements with reset, which span both presence and absence of unitality and contractivity of monitoring protocols, and illustrate the different mechanisms through which monitoring modifies the reservoir dynamics.

As a viable QRC map must be both strictly contractive and non-unital, this excludes unitary (unmonitored) dynamics. For any QRC dynamics $\Lambda$ in the absence of monitoring (last row of Table \ref{table:monitored-qrc}), viable designs require supplementing unitary evolution with information erasure \cite{FujiiNakajima2017, chen-nurdin-yamamoto, Fry2023} or non-unital noise channels \cite{sannia, monzani2025nonunitalnoisesuperconductingquantum, luca}, that induce the necessary contractivity. In the following section, we show how monitoring can itself serve as such a resource, rendering viable two dynamical evolutions that are otherwise unsuitable for QRC.

\subsection{Monitoring as a resource for QRC}\label{subsec:monitoring-resource}

Both unitary and unital-noise reservoir dynamics can become viable due to monitoring, depending on the measurement choice. Measurement back-action can provide the contractivity and non-unitality required for reservoir computing. In particular, both amplitude damping and partial projective measurement with reset lead to a viable QRC operation (\textit{Unitary} and \textit{Unital noise} columns in Table \ref{table:monitored-qrc}).

Representative monitored dynamics are shown in Fig. \ref{fig:dynamics_ESP_input} for the unitary evolution $\Lambda(s_k)[\cdot]=U(s_k)\cdot U^\dagger(s_k)$ with input encoding 
\begin{equation}\label{eq:unitary}
    U(s_k) = U_\text{ev}\Big(U_\text{enc}(s_k)^{\otimes N_E} \otimes \mathbb{I}^{\otimes(N - N_E)}\Big),
\end{equation}
where $U_\text{enc}(s_k) = \text{RY}(\phi(s_k))$ can span a subset $E$ of qubits via $\phi(s_k)=2\arcsin\sqrt{s_k}$, and $U_\text{ev}$ is a fixed unitary acting over all reservoir qubits. Throughout this work, the input is encoded into a single qubit ($N_E = 1$) and the unitary $U_\text{ev}$ is chosen as either a Haar-random unitary or a transverse-field Ising Hamiltonian evolution (see Appendix \ref{app:benchmark-reservoir} for more details). The Haar-random unitaries provide a generic benchmark for highly scrambling dynamics \cite{hayden2007black}, while the Ising evolution offers a physically motivated model with tunable correlations \cite{dynamical-phase-transitions}. 

Fig. \ref{fig:dynamics_ESP_input} illustrates the ESP (left panels) and separability (right panels) for different initial conditions and input sequences, respectively, in the unmonitored case and under different measurement schemes. For this purpose, here we consider a $N=5$ qubit system (reservoir), with $U_\text{ev}$ evolution for the Ising Hamiltonian  defined in Appendix \ref{app:benchmark-reservoir}. Additional qubits are used to perform the indirect measurements.

\begin{figure*}
    \centering
    \includegraphics[width=1.0\linewidth]{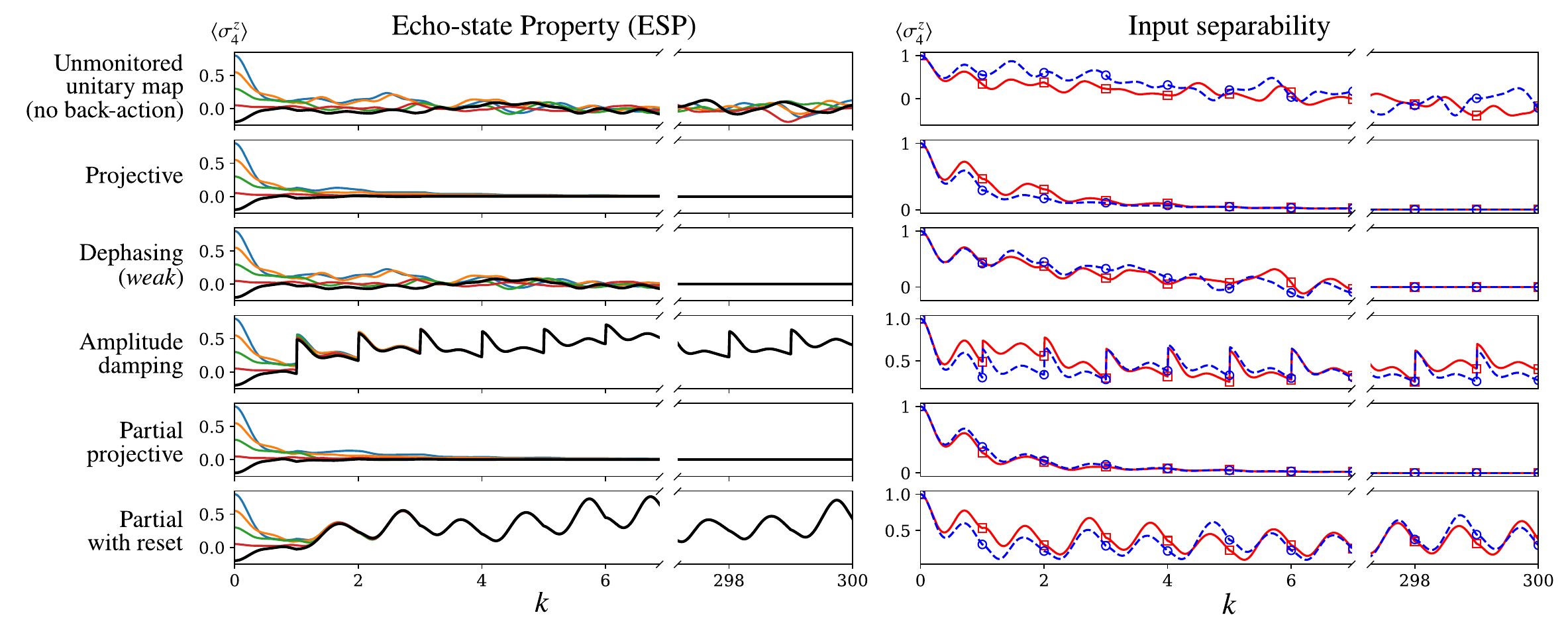}
    \caption{Dynamics of $\langle \sigma_4^z \rangle$ for the input-dependent unitary map in Eq. (\ref{eq:unitary}) with disordered Ising-Hamiltonian evolution under different measurement schemes. The left panels show the dynamics starting from five different initial reservoir states under the same input sequence, while the right panels show the dynamics generated by two different input sequences with a fixed initial state ($\ket{0}\!\!\bra{0}$). Unmonitored unitary maps (\textit{no back-action}) are not suitable for quantum reservoir computing, as they fail to satisfy the echo-state property (ESP). \textit{Projective}, \textit{dephasing}, and \textit{partial projective} measurement schemes are also unsuitable, since although they satisfy the ESP, the asymptotic state is input independent. In contrast, \textit{Amplitude damping} measurements and \textit{partial projective} measurements \textit{with reset} satisfy both the ESP and input separability, thereby providing the necessary dynamical conditions for effective QRC. As discussed in the main text, the partial projective measurement with reset is equivalent to an erase-and-write map. Similar results are found for the rest of $\sigma^z_i$ and their higher-order correlations.}
    \label{fig:dynamics_ESP_input}
\end{figure*}

For unitary (unmonitored) evolution, the dynamics reflects an initial state dependence, not displaying ESP. As indicated in Table \ref{table:monitored-qrc}, non-selective projective or dephasing measurements [Eq. (\ref{eq:deph-meas-ops})] are not strictly contractive, and neither is the composition of unitary dynamics with these measurement maps. Interestingly, the repeated composition becomes strictly contractive, as we anticipated in the previous subsection: contractivity can emerge in the composition even if absent in constituent maps. Indeed, the ESP is satisfied (see Appendix \ref{app:strict}). Still, these measurement schemes are unital and therefore yield input-independent fixed points, thereby preventing input separability \cite{thermalization-monitoring}, as shown in Fig. \ref{fig:dynamics_ESP_input} right panels. The same behavior appears if the system, evolving unitarily, is measured projectively in some of its qubits (partial projective panels in Fig. \ref{fig:dynamics_ESP_input}, where the first $N_P=3$ qubits of the $N$-qubit reservoir are measured). In contrast, a non-selective amplitude damping (AD) measurement is both strictly contractive and non-unital, analogously to unmonitored QRC maps with AD dissipation \cite{domingo_noise_qrc}, so that the monitored map satisfies the ESP and input separability for any input-dependent unitary.

The composition of unitary dynamics with a partial measurement with reset can also satisfy the necessary requirements for a viable QRC map \cite{hu_et_al} (partial with reset panels in Fig. \ref{fig:dynamics_ESP_input}, where, again, $N_P=3$ and $N_E=1$). For this measurement scheme, not all unitaries $U(s_k)$ generate suitable dynamics. To understand why, let us recall the structure of the partial measurement with reset: the system is partitioned into two subsets of qubits, $P$, which are projectively measured and reset, and $M$, which are not. As shown in Section \ref{subsec:partial}, this scheme is equivalent to an indirect measurement of $M$ by the probe qubits $P$, with measurement operators determined by the unitary $U(s_k)$ as
\begin{equation}
    \Omega_{\mathbf{b}_m}(s_k) = \bra{\mathbf{b}_m}_P U(s_k)\ket{0\cdots0}_{\!P}.
\end{equation}
Because of the reset, the evolution of $M$ corresponds to an open quantum system in a collisional model with governing map,
\begin{equation}
    \mathcal{T}'(s_k)\rho_{k-1}^M = \mathcal{M}_M(s_k)\rho_{k-1}^M=\sum_m\Omega_{\mathbf{b}_m}(s_k)\rho_{k-1}^M\Omega^\dagger_{\mathbf{b}_m}(s_k)
\end{equation}
where the input is encoded through the measurement process. The requirement on $U(s_k)$ is therefore that the reduced map $\mathcal{T}'(s_k)=\mathcal{M}_M(s_k)$ be strictly contractive with an input-dependent fixed point, which can be verified through the spectral properties of the map \cite{hu_et_al}. This can be ensured provided the unitary is sufficiently scrambling such that information is spread across all reservoir qubits.

We contrast two opposite situations with examples. First, a suitable construction is that of Eq. (\ref{eq:unitary}) with $U_\text{ev}$ chosen as a Haar-random unitary or an ergodic many-body Hamiltonian evolution \cite{dynamical-phase-transitions}. However, other constructions may fail. An example is the family of unitaries of the form
\begin{equation}
    U(s_k)=\left(\prod_{i=1}^{N/2}U_{P_i,M_i}^{(\text{deph})}\right)(\mathbb{I}^{\otimes N/2}\otimes U_M(s_k)),
\end{equation}
where $U_{P_i,M_i}^{(\text{deph})}$ couples qubit $i$ of $P$ with qubit $i$ in $M$ as in Eq. \eqref{eq:unitary-deph}, $U_M(s_k)$ is an encoding unitary acting on $M$, and $N_M=N_P=N/2$. This construction is equivalent to an input-dependent unitary evolution with dephasing noise on a subset of qubits in $M$, which was shown above to yield a reservoir that does not satisfy input separability. Partial measurement schemes with posterior reset thus shift the focus from measurement engineering of $M$ to reservoir dynamics engineering of $M$+$P$, both approaches being fundamentally equivalent as shown in Sec. \ref{sec:monitoring-qubits}.

Certain monitoring can make QRC not only viable but also influence performance. We analyze this aspect for the monitored unitary QRC map with AD monitoring in all qubits (see Section \ref{subsec:AD}) and unitary evolution as in Eq. (\ref{eq:unitary}) with an Ising Hamiltonian. The AD measurement features a tunable strength $\theta$ that can be optimized for task performance \cite{sannia}. Figure \ref{fig:stm_AD_meas} shows the performance on the linear short-term memory (STM) task, where the reservoir is trained to reconstruct a past input $y_k = s_{k-\tau}$ at delay $\tau$ of a 5-qubit monitored reservoir. At $\theta=\pi$, the measurement corresponds to a projective measurement followed by a reset, which erases all memory of past inputs and restricts the reservoir to learning only the current one. For finite-strength measurements ($0<\theta<\pi)$, a finite memory of past inputs is retained. For weak measurements ($\theta\to 0$), the long-term memory is enhanced at the expense of reducing the performance for small delays. Conversely, strong back-action (lines corresponding to $0.4\lesssim\theta<\pi$)  enhances the linear short-term memory of the reservoir while suppressing sensitivity to distant past inputs. These results display the sensitivity of the performance on the map even when ESP and separability are fulfilled and are compatible with the relation between quantum coherences and reservoir memory found in \cite{palacios-coherences}.

\begin{figure}
    \centering
    \includegraphics[width=1.0\linewidth]{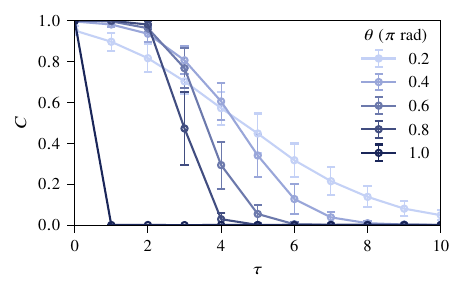}
    \caption{Correlation between the predicted and target time series for the linear short-term memory (STM) task for different measurement strengths $\theta$ of the Amplitude Damping measurement channel following the unitary unmonitored map defined in Eq. (\ref{eq:unitary}) with Ising Hamiltonian evolution. We recall that $\theta=\pi$ corresponds to a projective measurement with posterior reset, while $\theta\to 0$ represents a weak measurement. Error bars indicate the standard deviation obtained with 100 random reservoir parameters and random input series.}
    \label{fig:stm_AD_meas}
\end{figure}

Monitoring can thus serve as a resource for fulfilling the necessary conditions for QRC. In particular, we have shown  two monitoring schemes that lead to suitable QRC maps. These are the AD monitored QRC map, which is equivalent to unitary maps with an amplitude damping noise channel \cite{sannia,palacios-coherences}, and the partial monitoring with reset QRC map \cite{hu_et_al}. The latter is formally analogous to an erase-and-write map \cite{FujiiNakajima2017}, which will be treated in detail in the next subsection. Indeed, known suitable QRC maps reported in the literature (but not addressing a clear monitoring protocol) can be used as blueprints for designing monitoring schemes that enable real-time online execution of QRC protocols. To provide an example, any QRC map that has reset operations can naturally be made an online monitored QRC map by incorporating projective measurements before the reset, as this has no impact on the dynamics of the reservoir.

\subsection{Comparison of measurement schemes} \label{subsec:erase-write}

While monitoring can be a resource, making an evolution map suitable for QRC as seen in the previous subsection, some designs and implementations can display ESP and input separability even prior to monitoring being accounted for \cite{mujal_weak, chen-nurdin-yamamoto} as shown in Table \ref{table:monitored-qrc}, last row. This lifts the constraints on the measurement schemes that are unsuitable when the unmonitored dynamics is unitary, such as dephasing measurements, an approach that has been explored by monitoring a noisy unitary evolution projectively \cite{yasuda} and an erase-and-write map with dephasing measurements \cite{mujal_weak}. The unified framework of monitored QRC introduced above enables a systematic comparison of the effect of different measurement schemes on task performance for a common unmonitored dynamics, providing a tool for selecting the optimal scheme for a given task. Any suitable unmonitored QRC dynamics can be complemented with a monitoring map resulting in a suitable monitored QRC (see Table \ref{table:monitored-qrc}). However, the effect that measurement back-action has on task performance must be studied in detail for each reservoir.

As an example, here we consider an erase-and-write map, which has been extensively studied with different quantum systems \cite{FujiiNakajima2017, NokkalaGaussian, llodra, llodra2, prl_syk}. This model introduces dissipation by resetting a subset $E$ of reservoir qubits before each input injection. A widely used special case \cite{FujiiNakajima2017} takes $E$ to be a single qubit, writes the input into that qubit via $\ket{\psi(s_k)}=\sqrt{1-s_k}\ket{0}+\sqrt{s_k}\ket{1}$, and a scrambling unitary $U_{\text{ev}}$:
\begin{equation}\label{eq:erase-and-write-2}
 \Lambda(s_k)\rho_k = U_{\text{ev}}\bigl(\dyad{\psi(s_k)}\otimes \Tr_E[\rho_k]\bigr) U_{\text{ev}}^\dagger.
\end{equation}

The partial reset mechanism introduces dissipation and ensures that the reservoir satisfies the echo-state and fading-memory properties. In addition, performing a measurement on the reset qubits before reset does not alter the dynamical map. The partial measurement with reset is therefore dynamically equivalent to the erase-and-write map, with the measurement providing a means to implement an online erase-and-write protocol without introducing any additional dissipation.

\begin{figure}
    \centering
    \includegraphics[width=1.0\linewidth]{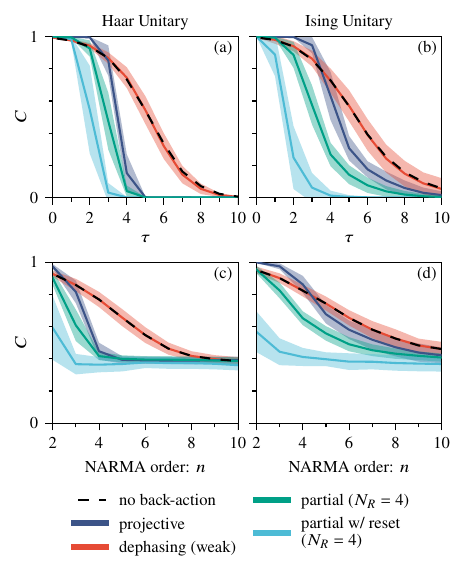}
    \caption{Correlation between the predicted and target time series for (a,b) the linear short-term memory (STM) task at different delays $\tau$, and (c,d) the NARMA$n$ task, using the monitored erase-and-write map with $N_E = 1$. Each panel compares the performance obtained with different measurement schemes $\mathcal{M}$. Panels (a,c) show results averaged over 100 realizations of Haar-random unitaries and input sequences, while panels (b,d) correspond to an Ising Hamiltonian evolution, with results averaged over 100 realizations of the reservoir parameters and input sequences. In all cases, the feature vector is constructed from the expectation values $\langle \sigma_i^z \rangle$, for $i=1,\dots,N$, with $N=5$, together with all higher-order correlation terms. The weak measurement results are obtained using a dephasing measurement with strength $\theta = 0.9\pi/2$. Shaded regions indicate one standard deviation around the mean.}
    \label{fig:stm_narma_qrc}
\end{figure}

Fig. \ref{fig:stm_narma_qrc} shows the performance of the monitored erase-and-write reservoir on the STM and NARMA$n$ tasks (see Appendix \ref{app:benchmark-reservoir}), under different monitoring schemes: $z$-projections of all qubits, dephasing measurements with $\theta=0.9\pi/2$ on all qubits (leading to weak monitoring), and partial projective measurements on $N_P=4$ of 5 qubits with and without posterior reset. We highlight two main aspects. First of all, for both benchmarks and for small delays $\tau$ ($=1,2$ e.g.) or low order $n$ (up to 3-4 for the Ising evolution), projective measurements on all reservoir qubits yield the highest performance, even surpassing the unmonitored case, with no back-action. This indicates that measurement-induced disturbance can enhance the computational capabilities of the reservoir \cite{mujal_weak, palacios-coherences, memorynonlinearitytradeoffquantumreservoir, franceschetto}. However, for these kinds of measurements, this performance advantage diminishes or even disappears for tasks requiring long-term memory. In this regime, measurement back-action becomes detrimental, and weak dephasing measurements yield the best performance, closely matching the results obtained in the absence of back-action.

A second observation is that partial measurements do not outperform the full projective measurement scheme in either benchmark, and the partial measurement with reset consistently yields the poorest performance across all memory regimes. This can be attributed to the combined effect of a reduced feature space and the strong erasure of information induced by the reset operation. It should be noted, however, that the indirect measurement schemes require additional auxiliary qubits, whereas the partial measurement with reset is implemented directly without this overhead.

\subsection{Impact of finite experimental realizations}\label{subsub:finite}

The relative performance of different monitoring strategies is not necessarily preserved in the presence of sampling noise, as we address here. In the idealized infinite-resource scenario, where an infinite number of trajectories are used to construct the features, all measurement schemes that probe the same observable would yield identical expectation values. However, these schemes cannot be regarded as equivalent, as the measurement back-action on the system generally differs between them in the limit of large samples. The impact of a finite number of quantum trajectories, $N_\text{shots}$, used to estimate the outcome probabilities, can be analyzed by adding a stochastic contribution to the expectation values obtained in the infinite-resources limit
\begin{equation}\label{eq:obs-finite-res}
    \langle O_i\rangle_\text{finite}\approx\langle O_i\rangle_\text{infinite} + \xi ,
\end{equation}
where $\xi \sim \mathcal{N}(0, \Delta\langle O_i\rangle_\text{max})$ \cite{mujal_weak}. The uncertainties $\Delta\langle O_i\rangle_\text{max}$, which correspond to their state-independent value or, equivalently, their maximum value, depend on both the observable considered and the measurement scheme, as discussed in Appendix \ref{app:finite}. The state-independent uncertainty overestimates the actual error, but provides an upper bound for the error. For the results of the finite-resource scenario, we consider that the feature vector consists only of expectation values of $\sigma^z$ and second-order correlations $\sigma^z\otimes\sigma^z$.

\begin{figure}
    \centering
    \includegraphics[width=1.0\linewidth]{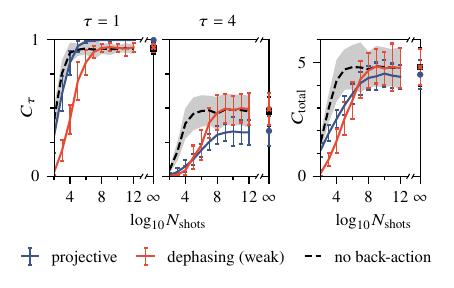}
    \caption{Effect of the finite resources on the Short-Term Memory (STM) capacity for the monitored erase-and-write map. Left panels show the correlation between the target and predicted sequences for delay $\tau=1$ and $\tau=4$ in the STM task. Right panel shows the total capacity computed as $C_\text{total} = \sum_{\tau=0}^{29}C_\tau$. Error bars and shaded region indicate one standard deviation obtained with 100 realizations of the Ising Hamiltonian evolution reservoir with random input series. Here, the features consist of only $\langle\sigma^z_i\rangle$ and $\langle\sigma^z_i\otimes\sigma^z_j\rangle$, and the dephasing (weak) measurement corresponds to $\theta=0.9\pi/2$.}
    \label{fig:stm-shotnoise-ising}
\end{figure}

Projective measurement schemes are more robust to shot noise than weak measurements, since the latter extract less information per shot \cite{mujal_weak}. For the STM task at short delay (Fig. \ref{fig:stm-shotnoise-ising}, $\tau=1$), where projective measurements give the best performance in the infinite-shot limit, saturating performance requires fewer experimental realizations with projective measurements, $\sim 10^6$, than with weak measurements, $\sim 10^8-10^9$. However, the impact of strong measurement back-action on the dynamics can be detrimental for the long-term memory of the reservoir, resulting in weak monitoring outperforming projective measurements, as previously shown. This behavior persists in the finite-shot regime, where the interplay between back-action and measurement statistics reflects the fundamental trade-off between information extraction and measurement-induced disturbance. Consequently, as the number of realizations increases, the robustness of projective measurements is insufficient to sustain better performance in long-term memory tasks, and weak measurements become preferable even in the finite-resource regime. As shown in Fig. \ref{fig:stm-shotnoise-ising}, the weak monitoring scheme outperforms the projective one at $N_\text{shots}\gtrsim 10^5$ for the STM task at delay $\tau=4$ (left panel) and at $N_\text{shots}\gtrsim 10^7$ for the total STM capacity (right panel).

The optimal measurement scheme under finite resources thus depends not only on measurement strength but on the specific task. These results are not restricted to the monitored erase-and-write map with an Ising Hamiltonian, and similar behavior is obtained with Haar-random evolution $U_\text{ev}$ and with an AD-monitored unitary map section \ref{subsec:monitoring-resource} (see Appendix \ref{app:finite}). 

\section{Time-multiplexing and measurement back-action}\label{sec:time-mult}

\begin{figure*}
    \centering
    \includegraphics[width=1.0\linewidth]{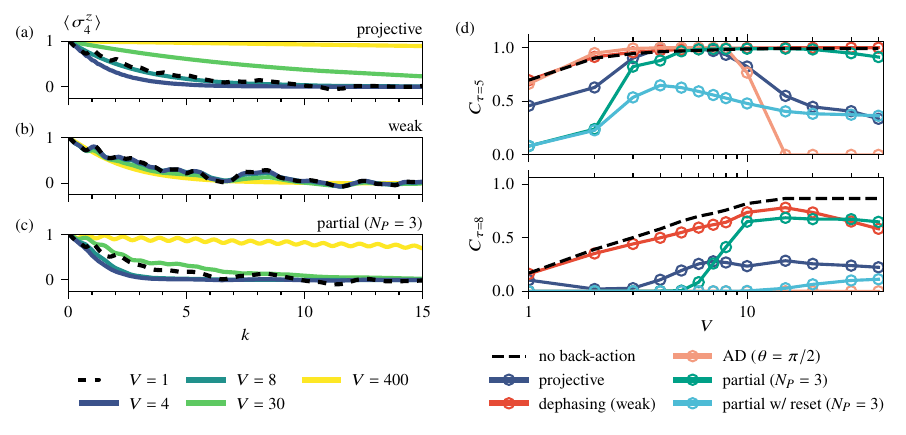}
    \caption{(a-c) Dynamics of $\langle\sigma^z_4\rangle$ under time-multiplexing. $V$ measurements are performed between consecutive input injections $k=1,2,\dots$ at intervals of $\Delta t/V$ for QRC, with an Ising reservoir operated as an erase-and-write map and under different monitoring schemes. Increasing $V$ introduces more back-action, effectively intensifying the effect of measurement dissipation. At the same time, the higher measurement rate can slow the unmonitored dynamics for large $V$, freezing the evolution as $V\to\infty$. (d) Impact of time-multiplexing on the short-term memory capacity $C$ at delay $\tau=5$ and $\tau=8$, applying time-multiplexing with $V$ virtual nodes for different measurement schemes. The weak dephasing measurement is set to a strength $\theta=0.9\pi/2$.}
    \label{fig:stm-timemultiplexing-ising}
\end{figure*}

In addition to measurement strength, another key resource is the dimensionality of the extracted features, which can be enhanced through time-multiplexing \cite{Appeltant2011, FujiiNakajima2017, QRC_IPC}. Originally introduced to enable classical reservoir computing with a single dynamical node, this technique involves performing stroboscopic measurements on the node during a fixed-input evolution. By extracting additional temporal features from the system dynamics, time-multiplexing effectively emulates a higher-dimensional reservoir and can significantly improve performance for classical reservoirs.

In the context of QRC, however, time-multiplexing amplifies the effect of the information-disturbance trade-off. Increasing the number of measurements augments the number of features $V$, referred to as virtual nodes, extracting more information from the reservoir and potentially enhancing performance. At the same time, for online monitored protocols, this leads to greater disturbance of the reservoir state, potentially degrading its memory. In this case, the strength and type of back-action can have a significant impact on the resulting task performance. A detailed study of time-multiplexing in monitored QRC is therefore essential.

The monitored QRC framework introduced above can be adjusted to accommodate time-multiplexing. Consider a single step of the monitored protocol between input injections which lasts a time $\Delta t$. Time-multiplexing consists of performing $V$ stroboscopic measurements within this time interval, subdividing the original single-step map $\mathcal{T}$ into $V$ sub-maps each of duration $\Delta t/V$. The set of $V$ measurements forms the expanded feature vector for input step $k$, generalizing Eq. \eqref{eq:output-features}:
\begin{equation}
    x_{k,i}^{(v)} = \text{Tr}[\Omega_i^\dagger \Omega_i \Lambda_v(s_k)\rho_k^{(v-1)}].
\end{equation}
This is obtained from all measurement outcomes $i$ and $v=1,\dots,V$, where 
\begin{equation}
    \rho_k^{(v)}=\left[\prod_{j=1}^{\overset{v}{\longleftarrow}}(\mathcal{M}_j\circ{\Lambda_j})(s_k)\right]\rho_{k-1}^{(0)},
\end{equation}
and $\rho_{k}^{(V)}=\rho_{k+1}^{(0)}$. Here, $\Lambda_j$ and $\mathcal{M}_j$ denote the new unmonitored and measurement maps and, for continuous-time dynamics, are obtained by rescaling the original maps to the virtual node timescale such that their joint evolution lasts $\Delta t/V$.

We now study the effect of time-multiplexing in the erase-and-write map of Section \ref{subsec:erase-write} with an Ising Hamiltonian evolution $U_\text{ev}$ of duration $\Delta t$, neglecting the time required for the encoding $U_\text{enc}(s_k)$ and each measurement. The first unmonitored sub-map $\Lambda_1$ takes the form given by Eq. (\ref{eq:erase-and-write-2}), while $\Lambda_v\rho = U_\text{ev}\rho U_\text{ev}^\dagger$ for $v\neq1$, and the Hamiltonian evolution $U_\text{ev}$ is rescaled to duration $\Delta t/V$. The measurement map is kept fixed throughout the protocol. $\mathcal{M}_v=\mathcal{M}$ for all $v$.

As a consequence of this rescaling, time-multiplexing also entails a change in the measurement rate. As this rate increases, the back-action progressively inhibits the unmonitored dynamics. In the limit $V\to\infty$, the state evolution is effectively frozen, displaying clear signatures of a quantum Zeno effect \cite{PhysRevA.41.2295} (Fig. \ref{fig:stm-timemultiplexing-ising}a,c). In this regime, the system cannot thermalize \cite{thermalization-monitoring}, as frequent measurements prevent energy transport and entropy redistribution. Reducing the measurement strength disturbs the system less, allowing for maintaining a good performance for a larger time-multiplexing than in the projective case (Fig. \ref{fig:stm-timemultiplexing-ising}b).

Figure \ref{fig:stm-timemultiplexing-ising}d illustrates how measurement back-action affects the linear memory capacity $C_\tau$ of the reservoir, highlighting the trade-off between information extraction and measurement-induced disturbance as the number of virtual nodes $V$ increases. Unlike the restarting protocol (no back-action), where the performance saturates, all measurement schemes exhibit an optimal value of $V$ that maximizes $C_\tau$, with the optimum depending on the measurement scheme. Similar results are obtained for other delays $\tau$, with the optimum also depending on $\tau$. Time-multiplexing thus provides an additional tunable resource for enhancing task performance, where both the number of virtual nodes and the measurement strength can be jointly optimized for a given task.

\section{Conclusions}

Online monitoring of quantum systems intrinsically intertwines information extraction and state disturbance. In this work, we developed a general framework for online quantum reservoir computing based on monitored quantum dynamics, unifying previously proposed projective, weak, and partial measurement schemes within a common description grounded in indirect measurement theory. This framework establishes measurement back-action not as an obstacle to be circumvented, but as a controllable dynamical resource that can be engineered to endow a quantum reservoir with the properties required for temporal information processing. 

A central result is that monitoring can provide the effective dissipation and non-unitality necessary for the echo-state property, fading memory, and input separability, even when the underlying unmonitored dynamics is unsuitable for in-memory QRC (e.g., unitary or unital noise evolutions). In particular, we derived a necessary and sufficient criterion (Theorem \ref{thm:disjoint} in Appendix \ref{app:strict}) for strict contractivity that emerges from the repeated composition of monitoring and reservoir maps that are not individually contractive, delimiting which measurement schemes can and cannot supply the missing dissipation. From this perspective, measurements play a dual role: they simultaneously extract the features used for computation and shape the reservoir dynamics itself. Not all measurement schemes, however, facilitate the viability of QRC with unitary unmonitored dynamics. This observation places measurement engineering on equal footing with Hamiltonian design and dissipation engineering as a tool for constructing quantum reservoirs. Furthermore, this framework reveals a clear and exploitable trade-off: stronger back-action enhances short-term memory at the cost of long-term retention, while weaker measurements preserve memory at the expense of information gain per shot. The presence of sampling noise also enters as a monitoring-dependent key aspect.

We further assessed time-multiplexing in the online setting, showing a richer structure than in restarting protocols. Increasing the number of virtual nodes expands the feature space but simultaneously amplifies measurement-induced disturbance, leading to an optimal value that balances information extraction against back-action and eventually freezes the reservoir dynamics through the quantum Zeno effect. This reveals measurement strength and the number of virtual nodes as jointly tunable parameters for optimizing monitored QRC performance without modifying the physical reservoir.

The framework provides a general route for designing quantum reservoirs through controlled monitoring. The correspondence between measurement schemes and effective dissipative dynamics suggests that many tools developed in the contexts of quantum trajectories, open quantum systems, and measurement-based control can be repurposed for quantum machine learning. Looking for experimental implementations, the indirect measurement schemes considered here are already within reach of current superconducting platforms, where fast reset operations are available, making monitored QRC a concrete target in noisy intermediate-scale quantum devices. While some implementations of monitored reservoirs have already been performed, this work provides a theoretical framework for understanding the impact of the monitoring on the QRC dynamics and the tools to engineer suitable monitoring schemes. More broadly, the general correspondence established between monitoring schemes and dissipative QRC maps suggests that the design principles developed here may extend beyond the qubit setting, for instance, in continuous-variable setups.

\begin{acknowledgements}
    We acknowledge support from the Spanish State Research Agency, through the María de Maeztu project CEX2021-001164-M, funded by MICIU/AEI/10.13039/501100011033; through the CoQuSy project PID2022-140506NB-C21 and -C22 funded by MICIU/AEI/10.13039/50110001103 and by ERDF, EU; and through the QuantCom project CNS2024-154720, funded by MICIU/AEI/10.13039/501100011033 and co-funded by the European Union; the project is funded under the Quantera II program that has received funding from the EU’s H2020 research and innovation program under Grant Agreement No. 101017733, and from the Spanish State Research Agency (project CoQuaDis PCI2024-153446) funded by MICIU/AEI/10.13039/50110001103; MINECO through the QUANTUM SPAIN project, and EU through the RTRP - NextGenerationEU within the framework of the Digital Spain 2025 Agenda; and CSIC's Quantum Technologies Platform (QTEP). O. M.-S. also acknowledges support from Grant PREP2022-000094 funded by MICIU/AEI/10.13039/501100011033 and FSE+, and thanks Rodrigo Martínez-Peña, Luísa Toledo Tude, and Nathan Keenan for fruitful discussions.
\end{acknowledgements}

\appendix

\section{Finite resources}\label{app:finite}

In an idealized scenario, expectation values are computed assuming complete knowledge of the system state $\rho$ or, equivalently, an infinite number of experimental repetitions of the monitored dynamics (quantum trajectories). In this regime, all compatible measurement schemes that probe a given observable yield the same expectation value. In realistic settings, however, only a finite number of trajectories, $N_\text{shots}$, is available, leading to statistical fluctuations that must be accounted for when estimating expectation values. As a result, the statistical uncertainty in the estimate of the expectation value scales as $\Delta\langle O\rangle \sim1/\sqrt{N_\text{shots}}$.

In addition, the magnitude of these statistical fluctuations depends on the measurement scheme employed. Projective measurements are maximally informative about a given observable, as each outcome yields the largest amount of information per shot. By contrast, finite-strength measurements extract less information in a single run and therefore exhibit larger statistical uncertainty for a fixed number of shots. Nevertheless, weak measurements can be advantageous in situations where minimizing measurement back-action is essential.

The dephasing and amplitude damping measurements introduced in the main text are both compatible measurements of $\sigma^z$ and all $N$-qubit correlations constructed from it. Here, we derive expressions for the statistical uncertainty in estimating the expectation value of the single-qubit observable for both schemes, as well as the relation between outcome probabilities and the two-qubit correlator $\sigma^z_i\otimes\sigma^z_j$, together with its statistical uncertainty.

Using (\ref{eq:deph-meas-sigma}) and standard error propagation, the statistical uncertainty in the estimate of the expectation value of $\sigma^z$ for the dephasing measurement scheme can be computed as
\begin{equation}
    \Delta\langle\sigma^z\rangle^{(\text{deph})}=\frac{1}{\sqrt{N_\text{shots}}}\cdot\frac{\sqrt{1-\cos^2(\theta)\cdot\langle\sigma^z\rangle^2}}{\cos\theta},
\end{equation}
One can observe that $\Delta\langle\sigma^z\rangle^{(\text{deph})}\geq \Delta\langle\sigma^z\rangle$, with the equality being satisfied in the projective measurement limit ($\theta=0$). As anticipated, finite-strength measurements of $\sigma^z$ ($0<\theta<\pi/2$) are less informative per experimental shot than projective measurements. Nevertheless, in the ideal infinite-ensemble limit ($N_\text{shots}\to\infty$), both measurement schemes yield equivalent expectation values.

The expression for the expectation value of $\sigma^z$ in terms of the AD measurement outcome probabilities was given in Eq. (\ref{eq:error-sigmaz-ad}). Its corresponding statistical uncertainty can be computed as
\begin{equation}\label{eq:unc-ad}
    \Delta\langle\sigma^z\rangle^{(\text{AD})} = \frac{1}{\sqrt{N_\text{shots}}}\cdot\frac{\sqrt{1-\langle\sigma^z\rangle_A^2}}{\sin^2(\theta/2)}
\end{equation}
where $\langle\sigma^z\rangle_A=\text{Pr}_\rho(0)-\text{Pr}_\rho(1)$ and can be related to the system expectation value $\langle\sigma^z\rangle$ as 
\begin{equation}
    \langle\sigma^z\rangle_A=\sin^2(\theta/2)\langle\sigma^z\rangle+ \cos^2(\theta/2).
\end{equation}

For the dephasing measurement channel, the expectation values of the 2nd order correlators of $\sigma^z$ can be computed from the outcome probabilities as
\begin{equation}
    \langle\sigma^z_i\otimes\sigma^z_j\rangle = \frac{1}{\cos^2\theta}\langle\sigma^z_i\otimes\sigma^z_j\rangle_A ,   
\end{equation}
where $\langle\sigma^z_i\otimes\sigma^z_j\rangle_A= \text{Pr}_\rho(0_i,0_j) + \text{Pr}_\rho(1_i,1_j) - \text{Pr}_\rho(1_i,0_j) - \text{Pr}_\rho(0_i,1_j)$. The associated uncertainty in the estimate of the expectation value is therefore
\begin{equation}\label{eq:unc-deph-2}
    \Delta\langle\sigma^z_i\otimes\sigma^z_j\rangle^{(\text{deph})}=\frac{1}{\sqrt{N_\text{shots}}}\cdot\frac{\sqrt{1-\cos^4(\theta)\cdot\langle\sigma^z_i\otimes\sigma^z_j\rangle^2}}{\cos^2(\theta)}
\end{equation}
It can be observed that the impact of finite-strength measurements becomes increasingly pronounced for higher-order $\sigma^z$ correlations, resulting in larger statistical fluctuations in the finite-resources regime.

The second-order correlations for the AD measurement channel and its uncertainty can be computed as
\begin{widetext}
\begin{equation}
\langle \sigma^z_i \otimes \sigma^z_j \rangle = \frac{\langle \sigma^z_i \otimes \sigma^z_j \rangle_A - \cos^2(\theta/2)\left( \langle \sigma^z_i \rangle_A + \langle \sigma^z_j \rangle_A \right) + \cos^4(\theta/2)}{\sin^4(\theta/2)}
\end{equation}
\begin{equation}
\Delta\langle \sigma^z_i \otimes \sigma^z_j \rangle^{(\mathrm{AD})} = \frac{1}{\sqrt{N_{\mathrm{shots}}}} \cdot \frac{1}{\sin^4(\theta/2)} \left( \sqrt{1 - \langle \sigma^z_i \otimes \sigma^z_j \rangle_A^2} + \cos^2(\theta/2) \left( \sqrt{1 - \langle \sigma^z_i \rangle_A^2} + \sqrt{1 - \langle \sigma^z_j \rangle_A^2} \right) \right)
\label{eq:unc-ad-2}
\end{equation}
\end{widetext}
where, again, $\langle\sigma^z\rangle_A=\text{Pr}_\rho(0)-\text{Pr}_\rho(1)$ and $\langle\sigma^z_i\otimes\sigma^z_j\rangle_A= \text{Pr}_\rho(0_i,0_j) + \text{Pr}_\rho(1_i,1_j) - \text{Pr}_\rho(1_i,0_j) - \text{Pr}_\rho(0_i,1_j)$ are computed from the measurement statistics.

Since the ranges of the expectation values are $-1\leq\langle\sigma^z\rangle\leq1$ and $-1\leq\langle\sigma^z_i\otimes\sigma^z_j\rangle\leq1$, the maximum uncertainty is obtained when these expectation values are equal to 0. Therefore, one can consider the state-independent uncertainty as the upper bound of the previous uncertainties \cite{mujal_weak}, which is obtained by substituting $\langle\sigma^z\rangle=0$ and $\langle\sigma^z_i\otimes\sigma^z_j\rangle=0$ in Equations (\ref{eq:unc-ad})--(\ref{eq:unc-ad-2}). This corresponds to the $\Delta\langle O\rangle_\text{max}$ used to compute the expectation values in the finite resource scenario, as given by Equation (\ref{eq:obs-finite-res}).

Fig. \ref{fig:stm-shotnoise-AD} shows the impact on task performance of a finite number of realizations for the AD-monitored unitary map introduced in Sec. \ref{subsec:monitoring-resource}. As observed in Sec. \ref{subsub:finite}, the optimal measurement scheme depends on both the number of shots and the specific task: strong measurements, which extract more information per shot, are preferable for short-term memory tasks, while weak measurements become advantageous when long-term memory is required.

\begin{figure}
    \centering
    \includegraphics[width=1.0\linewidth]{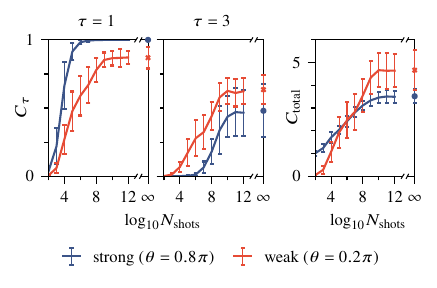}
    \caption{Effect of the finite resources on the Short-Term Memory (STM) capacity for the monitored unitary reservoir with AD monitoring (Section \ref{subsec:monitoring-resource}) in the weak and strong measurement regimes. Left panels show the correlation between the target and predicted sequences for delay $\tau=1$ and $\tau=3$. Right panel shows the total capacity computed as $C_\text{total} = \sum_{\tau=0}^{29}C_\tau$. Error bars and shaded region indicate one standard deviation obtained with 100 random realizations of the Ising Hamiltonian and random input series. Here, the features consist of only $\langle\sigma^z_i\rangle$ and $\langle\sigma^z_i\otimes\sigma^z_j\rangle$.}
    \label{fig:stm-shotnoise-AD}
\end{figure}

\section{QRC benchmark tasks, training and unmonitored dynamics}\label{app:benchmark-reservoir}

In order to assess the performance of the reservoirs, we use two benchmark tasks \cite{benchmarks}: the linear Short-Term Memory (STM) and the Normalized Auto-Regressive Moving Average of order $n$ (NARMA$n$) tasks, which assess the linear and nonlinear memory of the reservoir. The STM task is defined as:
\begin{equation}
    y_k = s_{k-\tau},
\end{equation}
where $\tau\in\mathbb{N}$ is the delay. The NARMA$n$ task, where $n\in\mathbb{N}$, is defined as:
\begin{equation}
    y_{k} = 0.3 y_{k-1} + 0.05 \sum_{\tau = 1}^{n} y_{k-1} y_{k-\tau} + 1.5 s_{k - n} s_{k-1} + 0.1
\end{equation}
Performance is quantified via the squared Pearson correlation coefficient between the target sequence $\mathbf{y}$ and the predicted sequence $\hat{\mathbf{y}}$,
\begin{equation}\label{eq:capacity}
    C=\frac{\text{cov}^2(\mathbf{y}, \hat{\mathbf{y}})}{\sigma(\mathbf{y})^2\sigma(\hat{\mathbf{y}})^2},
\end{equation}
where $\text{cov}$ denotes the covariance and $\sigma$ the standard deviation. This metric, often referred to as the \textit{capacity}, satisfies $C=1$ when the reservoir perfectly predicts the target and $C=0$ when it fails completely.

In both cases, a length $L=3000$ input sequence has been considered. The first 1000 steps are washed out, such that the reservoir is able to forget its initial state. Steps $k\in\mathcal{K}_\text{train}=[1000,1999]$ are used for training, which is performed by optimizing Eq. \eqref{eq:output} via linear regression. The optimized weight vector $\mathbf{W}$ is thus computed as
\begin{equation}
    \mathbf{W}=\mathbf{y}_\text{train}\mathbf{X}_\text{train}^+,
\end{equation}
where $\mathbf{y}_\text{train}=\{y_k \mid k\in\mathcal{K}_\text{train}\}$ is a vector containing the target training sequence, $\mathbf{X}_\text{train}^+$ is the Moore-Penrose pseudo-inverse of the training feature matrix $\mathbf{X}_\text{train}=\{\mathbf{x}_k\mid k\in\mathcal{K}_\text{train}\}$. Finally, the last 1000 steps are used to test the performance of the reservoir by computing the correlation between the target and the output sequences through Eq. \eqref{eq:capacity}. For the STM task, the input is chosen as a uniform random sequence, $\mathbf{s}=\{s_k | s_k \sim U[0,1]\}$. For the NARMA$n$ task, the input is rescaled by a factor $c=0.2$ in order to avoid divergences.

In the time-multiplexing scenario, the dimension of the feature vector increases, and the previously used training-sequence length is no longer statistically sufficient to prevent overfitting. Therefore, in this case, the input sequence and training phase are extended to include a larger amount of data (see Table \ref{tab:params}).

\renewcommand{\arraystretch}{1.}
\begin{table}
\centering
\begin{tabular}{l|c}
\textbf{Parameter} & \textbf{Value} \\
\hline
\hline
Number of system qubits $N$ & 5 \\
\hline
Ising Hamiltonian evolution:\\
$h$ & 1.0 \\
$J$ & 1.0 \\
$\Delta t$ & 2.0 \\
$\Delta t'$ (time-multiplexing) & $\Delta t / V$ \\
\hline
STM input & $s_k \sim U[0,1]$ \\
NARMA input & $s_k \sim U[0,0.2]$ \\
\hline
Total length $L$ (standard) & 3000 \\
Washout & $[0,1000)$ \\
Training & $[1000,2000)$ \\
Testing & $[2000,3000)$ \\
\hline
Total length $L$ (multiplexing) & 14000 \\
Washout & $[0,1000)$ \\
Training & $[1000,11000)$ \\
Testing & $[11000,14000)$ \\
\hline
\end{tabular}
\caption{Simulation parameters.}
\label{tab:params}
\end{table}

As discussed in the main text, the evolution unitaries $U_\text{ev}$ [Eqs. \eqref{eq:unitary} and \eqref{eq:erase-and-write-2}] considered in this paper are Haar-random unitaries and Ising Hamiltonian dynamics. These unitaries can provide the scrambling needed for suitable QRC dynamics. Haar-random unitaries are sampled from the Haar measure and kept fixed throughout each simulation. For the Ising Hamiltonian reservoirs, the unitary evolution is constructed as
\begin{equation}
    U_\text{ev}=e^{-iH_\text{Ising}\Delta t}
\end{equation}
assuming $\hbar=1$ and the Hamiltonian $H_\text{Ising}$ is a fully-connected Transverse-field Ising Hamiltonian:
\begin{equation}\label{eq:ising_hamiltonian}
H = \sum_{i>j} J_{ij}\sigma_i^x\sigma_j^x + h\sum_{i}\sigma_i^z   
\end{equation}
Here, $J_{ij}$ are the coupling strengths (distributed randomly between $[-J, J]$), $h$ is the strength of the transverse field, and $\sigma^d_{i}$ with $d\in\{x,y,z\}$ are the Pauli matrices acting on qubit $i$ (i.e. $N$ qubits operators consisting in the tensor product of the $\sigma^d$ Pauli matrix acting on qubit $i$ and identities acting on the rest of the qubits). 

In all our simulations, we have considered $N = 5, J=1, h=1$ and $\Delta t=2$, which corresponds to an ergodic phase of the many-body system \cite{dynamical-phase-transitions}. When time-multiplexing is used, $\Delta t$ is rescaled to $\Delta t/V$ in Eq. \ref{eq:ising_hamiltonian}, where $V$ is the number of virtual nodes. All simulation parameters are summarized in Table \ref{tab:params}.

\section{Strict contractivity of the composition of two CPTP maps}\label{app:strict}

\begin{figure*}
    \centering
    \includegraphics[width=1.0\linewidth]{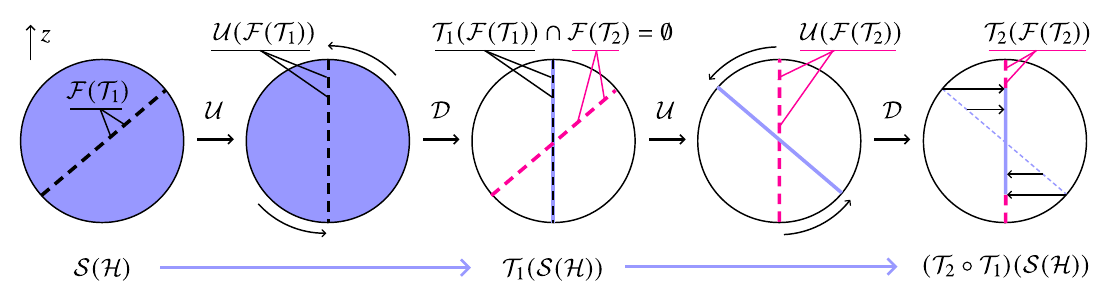}
    \caption{The composition $\mathcal{T}_2\circ\mathcal{T}_1$ of two CPTP maps can be strictly contractive, shrinking distances between all pairs of states, even when the individual maps are not. As an example, we consider the composition of two identical single-qubit maps, $\mathcal{T}_1=\mathcal{T}_2=\mathcal{D}\circ\mathcal{U}$, consisting of a unitary map $\mathcal{U}$ followed by a completely dephasing noise channel $\mathcal{D}$, which are not strictly contractive. The circle represents a cross-section of the single-qubit state space (Bloch sphere) $\mathcal{S}(\mathcal{H})$; the violet region indicates the evolving state space under each channel; and the black and pink dashed lines denote, respectively, the sets of state differences $\rho-\sigma$ that preserve distances under $\mathcal{T}_1$ and under $\mathcal{T}_2$ (denoted as $\mathcal{F}(\mathcal{T}_1)$ and $\mathcal{F}(\mathcal{T}_2)$), evolved by each channel. Under composition, the input of $\mathcal{T}_2$ is restricted to the images of $\mathcal{T}_1$. Therefore, if no pair of states that preserve distances under $\mathcal{T}_2$ is within the $\mathcal{T}_1$-images of pairs of states that preserve distances under $\mathcal{T}_1$ (equivalently, $\mathcal{T}_1(\mathcal{F}(\mathcal{T}_1))\cap\mathcal{F}(\mathcal{T}_2)=\emptyset$), the composition becomes strictly contractive.}
    \label{fig:strict-contractivity}
\end{figure*}

In this Appendix, we formalize the conditions under which the composition of two CPTP maps is strictly contractive, also considering the case (relevant for monitored QRC) where neither of the constituent maps is. In the following, $\mathcal{S}(\mathcal{H})$ denotes the set of density operators on a finite-dimensional Hilbert space $\mathcal{H}$, and $\kappa^{\max}(\mathcal{T})$ is the contraction coefficient of a CPTP map $\mathcal{T}$ defined in Eq. \eqref{eq:cont-coef}. Given that the ratio in Eq. \eqref{eq:cont-coef} depends only on the normalized traceless Hermitian difference $\Delta = \rho - \sigma$, and the set $\{\Delta = \Delta^\dagger : \mathrm{Tr}\,\Delta = 0,\, \|\Delta\|_1 = 1\}$ is compact, the supremum is attained in finite dimensions; we use this fact below.

\begin{proposition}[Submultiplicativity of the contraction
coefficient]\label{prop:submult}
Let $\mathcal{T}_1$, $\mathcal{T}_2$ be CPTP maps on $\mathcal{S}(\mathcal{H})$. Then
\begin{equation}
    \kappa^{\max}(\mathcal{T}_2 \circ \mathcal{T}_1)
    \;\le\; \kappa^{\max}_{2,\mathrm{rest}} \cdot \kappa^{\max}_1
    \;\le\; \kappa^{\max}_2 \cdot \kappa^{\max}_1 ,
\label{eq:submult}
\end{equation}
where $\kappa^{\max}_{2,\mathrm{rest}}$ denotes the contraction coefficient of $\mathcal{T}_2$ restricted to the image $\mathcal{T}_1(\mathcal{S}(\mathcal{H})) \subseteq \mathcal{S}(\mathcal{H})$. In particular, if at least one of the two maps is strictly contractive, so is the composition.
\end{proposition}

\begin{proof}
For any pair $\rho \neq \sigma$ with $\mathcal{T}_1\rho \neq \mathcal{T}_1\sigma$,
\begin{widetext}
    \begin{equation}
    \frac{\|\mathcal{T}_2 \circ \mathcal{T}_1 \rho
          - \mathcal{T}_2 \circ \mathcal{T}_1 \sigma\|_1}
         {\|\rho - \sigma\|_1}
    =\\
    \frac{\|\mathcal{T}_2 \circ \mathcal{T}_1 \rho
          - \mathcal{T}_2 \circ \mathcal{T}_1 \sigma\|_1}
         {\|\mathcal{T}_1\rho - \mathcal{T}_1\sigma\|_1}
    \cdot
    \frac{\|\mathcal{T}_1\rho - \mathcal{T}_1\sigma\|_1}
         {\|\rho - \sigma\|_1},
    \end{equation}
\end{widetext}
while pairs with $\mathcal{T}_1\rho = \mathcal{T}_1\sigma$ contribute zero to the supremum. Taking suprema of the two factors separately yields the first inequality in Eq. \eqref{eq:submult}; the second follows from enlarging the restricted supremum to all of $\mathcal{S}(\mathcal{H})$.
\end{proof}

Proposition \ref{prop:submult} covers the case discussed in Sec. \ref{sub:online-monitored-qrc} [Eq. \eqref{eq:strict-composition}], in which either the unmonitored dynamics $\Lambda$ or the measurement map $\mathcal{M}$ is strictly contractive. On the other hand, the converse is not true: strict contractivity of the composition does not require it of either constituent. To discuss when it emerges, let us introduce the following notion.

\begin{definition}[Saturating set]\label{def:saturating}
The \emph{saturating set} of a CPTP map $\mathcal{T}$ is the set of pairs of states whose distance is preserved by $\mathcal{T}$,
\begin{align}
    \mathcal{F}(\mathcal{T})=\biggl\{(\rho,\sigma)&\in\mathcal{S}(\mathcal{H})\times\mathcal{S}(\mathcal{H}) :  \nonumber \\
    & \rho\neq\sigma, \frac{\norm{\mathcal{T}\rho-\mathcal{T}\sigma}_1}{\norm{\rho-\sigma}_1}=1\biggr\},
\end{align}
and we write $\mathcal{T}(\mathcal{F}(\mathcal{T})) = \{(\mathcal{T}\rho, \mathcal{T}\sigma) : (\rho,\sigma) \in \mathcal{F}(\mathcal{T})\}$ for the set of its image pairs. A map is strictly contractive if and only if $\mathcal{F}(\mathcal{T}) = \emptyset$.
\end{definition}

\begin{theorem}[Disjointness criterion for emergent strict contractivity]
\label{thm:disjoint}
Let $\mathcal{T}_1$, $\mathcal{T}_2$ be CPTP maps on a finite-dimensional Hilbert space. The composition $\mathcal{T}_2 \circ \mathcal{T}_1$ is strictly contractive if and only if
\begin{equation}
    \mathcal{T}_1(\mathcal{F}(\mathcal{T}_1)) \cap \mathcal{F}(\mathcal{T}_2) = \emptyset .
\label{eq:disjointness}
\end{equation}
\end{theorem}

\begin{proof}
Suppose first that $\mathcal{T}_2 \circ \mathcal{T}_1$ is not strictly contractive, i.e.\ $\kappa^{\max}(\mathcal{T}_2 \circ \mathcal{T}_1) = 1$. Since $\mathcal{S}(\mathcal{H})$ is compact and the ratio in Eq. \eqref{eq:cont-coef} is continuous in $(\rho,\sigma)$, the supremum is attained; hence there exists a pair $(\rho,\sigma)$, $\rho \neq \sigma$, with $\|\mathcal{T}_2 \circ \mathcal{T}_1 (\rho - \sigma)\|_1  = \|\rho - \sigma\|_1$. Because CPTP maps are contractive in trace norm, the chain
\begin{equation}
\|\mathcal{T}_2 \circ \mathcal{T}_1 (\rho - \sigma)\|_1 \le \|\mathcal{T}_1 (\rho - \sigma)\|_1 \le \|\rho - \sigma\|_1
\end{equation}
must be saturated at both steps. The second equality gives $(\rho,\sigma) \in \mathcal{F}(\mathcal{T}_1)$ (and, in particular, $\mathcal{T}_1\rho \neq \mathcal{T}_1\sigma$), while the first gives $(\mathcal{T}_1\rho, \mathcal{T}_1\sigma) \in \mathcal{F}(\mathcal{T}_2)$. Hence the intersection in Eq. \eqref{eq:disjointness} is nonempty.

Conversely, suppose the intersection is nonempty; then there exists $(\rho,\sigma) \in \mathcal{F}(\mathcal{T}_1)$ with $(\mathcal{T}_1\rho, \mathcal{T}_1\sigma) \in \mathcal{F}(\mathcal{T}_2)$. Composing the two distance-preserving steps gives $\|\mathcal{T}_2 \circ \mathcal{T}_1 \rho - \mathcal{T}_2 \circ \mathcal{T}_1 \sigma\|_1 = \|\rho - \sigma\|_1$, so $\kappa^{\max}(\mathcal{T}_2 \circ \mathcal{T}_1) = 1$ and the composition is not strictly contractive.
\end{proof}

According to Theorem \ref{thm:disjoint}, strict contractivity of the composition fails precisely when some pair of states that preserves distances under $\mathcal{T}_1$ is mapped by $\mathcal{T}_1$ onto a pair that also preserves distances under $\mathcal{T}_2$. We now apply this criterion to the monitored unitary dynamics of the main text.

\medskip
\noindent\textbf{Unitary dynamics under dephasing monitoring.}
Consider a single qubit, a unitary channel $\mathcal{U}\rho = U\rho U^\dagger$, and the completely dephasing channel along the $z$-axis, $\mathcal{D}\rho = \langle 0|\rho|0\rangle\, |0\rangle\!\langle 0| + \langle 1|\rho|1\rangle\, |1\rangle\!\langle 1|$, equivalent to a non-selective projective measurement in the computational basis. Neither map is strictly contractive: $\mathcal{F}(\mathcal{U}) = \mathcal{S}(\mathcal{H})\times\mathcal{S}(\mathcal{H})$, since unitaries preserve all distances, while $\mathcal{F}(\mathcal{D})$ consists of all pairs whose difference $\rho - \sigma$ lies along the $z$-axis of the Bloch sphere --- in particular, all pairs of distinct elements of $\mathcal{Z}$, the set of diagonal states.

\begin{remark}[A single monitored step is not strictly contractive]\label{rem:single-step}
For the single composition $\mathcal{D}\circ\,\mathcal{U}$, one has $\mathcal{U}(\mathcal{F}(\mathcal{U})) = \mathcal{S}(\mathcal{H})\times\mathcal{S}(\mathcal{H})$, since unitary channels map the state space onto itself bijectively. Hence $\mathcal{F}(\mathcal{D}) \subset \mathcal{U}(\mathcal{F}(\mathcal{U}))$, the intersection in Eq. \eqref{eq:disjointness} is nonempty, and by Theorem \ref{thm:disjoint} the composition of a random unitary and a dephasing channel is not strictly contractive.
\end{remark}

\begin{corollary}[Emergent strict contractivity under repeated
monitoring]\label{cor:repeated}
Let $\mathcal{T}_1 = \mathcal{T}_2 = \mathcal{D}\circ\,\mathcal{U}$ describe two consecutive steps of the dephasing-monitored unitary dynamics. Then $\mathcal{T}_2 \circ \mathcal{T}_1$ is strictly contractive if and only if $\mathcal{U}^{-1}(\mathcal{Z}) \neq \mathcal{Z}$, i.e., if and only if $U$ is not a rotation about the $z$-axis.
\end{corollary}

\begin{proof}
The saturating sets of both (identical) maps coincide and contain exactly the pairs of states whose difference lies along the rotated axis $\mathcal{U}^{-1}(\mathcal{Z})$, where $\mathcal{U}^{-1}\rho = U^\dagger \rho\, U$: these are precisely the differences that the unitary maps onto the $z$-axis, along which the dephasing channel preserves the trace norm. The image set is easily characterized: $\mathcal{D}$ maps every state to a diagonal state, so $\mathcal{T}_1(\mathcal{F}(\mathcal{T}_1))$ consists of pairs of distinct diagonal states, whose differences all lie along the $z$-axis. Thus $\mathcal{T}_1(\mathcal{F}(\mathcal{T}_1)) \cap \mathcal{F}(\mathcal{T}_2) = \emptyset$ if and only if $\mathcal{U}^{-1}(\mathcal{Z}) \neq \mathcal{Z}$, and the claim follows from Theorem \ref{thm:disjoint}. See Fig. \ref{fig:strict-contractivity} for a geometric illustration.
\end{proof}

\begin{corollary}[ESP without input separability]\label{cor:esp-no-sep}
Under the conditions of Corollary \ref{cor:repeated}, the repeated map converges to a unique, input-independent fixed point,
\begin{equation}
    \lim_{n\to+\infty} (\mathcal{T}_2\circ\mathcal{T}_1)^n \rho_0 = \frac{\mathbb{I}}{d} \qquad \forall \rho_0 \in \mathcal{S}(\mathcal{H}),
\end{equation}
with $d = \dim(\mathcal{H})$.
\end{corollary}

\begin{proof}
Strict contractivity guarantees a unique fixed point \cite{raginsky}; since both sub-maps are unital, their only common fixed point is the maximally mixed state.
\end{proof}

\begin{remark}[Finite measurement strength]\label{rem:finite-strength}
These results are not restricted to completely dephasing maps but hold for arbitrary dephasing channels. A dephasing measurement of finite strength [Eq. \eqref{eq:deph-meas-ops}] multiplies the transverse components of the Bloch vector by $\gamma = \sin\theta < 1$ for $\theta \in [0,\pi/2)$, so that only differences along the $z$-axis preserve their trace norm: the saturating sets are unchanged and Corollaries \ref{cor:repeated} and \ref{cor:esp-no-sep} apply verbatim.
\end{remark}

Corollary \ref{cor:esp-no-sep} makes precise the statement of Sec. \ref{subsec:monitoring-resource}: unitary dynamics under non-selective dephasing monitoring satisfies the echo-state property (strict contractivity emerges from the repeated composition even though it is absent in each individual step) but fails input separability, since all input-dependent trajectories converge to the same maximally mixed state, making different inputs indistinguishable. This is why dephasing (and, by the same argument, partial projective) monitoring cannot render unitary reservoirs viable, whereas non-unital schemes such as amplitude-damping monitoring can (Table \ref{table:monitored-qrc}).

\bibliography{paper1_bib.bib}

\end{document}